\documentclass[a4paper,onecolumn,11pt,unpublished]{quantumarticle}
\pdfoutput=1
\usepackage{geometry}
\usepackage[english]{babel}
\usepackage{subcaption}
\usepackage{url}
\usepackage[normalem]{ulem}
\usepackage[utf8]{inputenc}
\usepackage[T1]{fontenc}
\usepackage{amsmath,amssymb,amsfonts,amsbsy,amsthm}
\usepackage{graphicx}
\usepackage{braket}
\usepackage{physics}
\usepackage{bm}
\usepackage{hyperref}
\usepackage{xcolor}
\usepackage{tikz}
\usetikzlibrary{quantikz2}
\usetikzlibrary{shapes.geometric}
\hypersetup{
  colorlinks=true,
  linkcolor=blue,
  citecolor=blue,
  urlcolor=blue
}
\usepackage{orcidlink}
\newtheorem{theorem}{Theorem}

\begin{document}
\title{Explicit Block Encodings of Discrete Laplacians with Mixed Boundary Conditions}

\author{Alexandre Boutot\,\orcidlink{0009-0002-9431-1864}}
\email{aboutot@uwaterloo.ca}
\affiliation{Department of Physics and Astronomy, University of Waterloo, Waterloo, ON, Canada, N2L 3G1}

\author{Viraj Dsouza\,\orcidlink{0009-0007-3512-9815}}
\email{virajdanieldsouza@gmail.com}
\affiliation{Quantum Computing Researcher, BQP, New York, USA}



\begin{abstract}
Discrete Laplacian operators arise ubiquitously in scientific computing and frequently appear in quantum algorithms for tasks such as linear algebra, Hamiltonian simulation, and partial differential equations. Block encoding provides the standard method for accessing matrix data within quantum circuits. Efficient implementations of such algorithms require efficient block encodings of the discretized operator. While several general-purpose techniques exist for block encoding arbitrary matrices, they usually require deep quantum circuits. Moreover, existing efficient constructions that exploit Laplacian structure are limited in scope, typically assuming fixed boundary conditions or uniform grid resolutions. In this work, we present a unified framework for efficiently block encoding finite-difference discretizations of the Laplacian that supports Dirichlet, periodic, and Neumann boundary conditions in arbitrary spatial dimensions. Our construction allows different boundary conditions and grid sizes to be specified independently along each coordinate axis, enabling mixed-boundary and anisotropic discretizations within a single modular circuit architecture. We provide analytical gate-complexity estimates and perform circuit-level benchmarks after transpilation to an IBM hardware gate set. Across one-, two-, and three-dimensional examples, the resulting circuits exhibit substantially lower gate counts and higher success probabilities when compared to certain existing approaches.

\end{abstract}

\maketitle


\section{Introduction}

The Laplacian operator plays a central role in a wide range of mathematical and physical models, appearing ubiquitously in partial differential equations (PDEs), spectral graph theory, quantum walks, quantum image processing and Hamiltonian simulation. In numerical settings, the Laplacian is typically discretized using finite-difference or finite-element schemes~\cite{smith1993numerical}, yielding structured, sparse matrices that reflects the dimensionality, grid resolution, and boundary conditions of the underlying domain. While the structured sparsity of the discrete Laplacian enables efficient classical solvers~\cite{krylov} in low dimensions, the size of the discretized operator grows as \(O(N^D)\) for a \(D\)-dimensional domain with \(N\) grid points per dimension. This exponential growth quickly becomes prohibitive for classical computers. 

Quantum computing~\cite{NielsenChuang2002} offers a fundamentally different approach by representing vectors and operators in exponentially large Hilbert spaces using only a polynomial number of qubits, thereby alleviating the memory bottleneck associated with explicitly storing large discretized operators. Within this paradigm, quantum algorithms for linear algebra~\cite{morales2025quantumlinearsolverssurvey} provide a natural framework for manipulating operators such as the Laplacian directly in their matrix form. A central primitive enabling these algorithms is block encoding~\cite{qspLow}, which allows a generally non-unitary matrix to be embedded into a larger unitary operator acting on an extended Hilbert space.

Given a generally non-unitary matrix \( A \), a block encoding embeds \( A \) into a larger unitary operator \( U \) acting on an extended Hilbert space such that
\begin{equation}
U =
\begin{pmatrix}
A / \alpha & * \\
* & *
\end{pmatrix},
\qquad \|A\| \le \alpha ,
\label{eq:block_encoding}
\end{equation}
where \( \alpha \) is the sub-normalization factor and the remaining blocks are irrelevant. Block encodings provide the standard input model for quantum singular value transformation (QSVT)~\cite{Gily_n_2019}, which unifies a broad class of quantum algorithms~\cite{camps2022fable} including Hamiltonian simulation, matrix inversion, phase estimation, unstructured search and the eigen value threshold problem. The efficiency of any such algorithm depends critically on the quality of the block encoding, as quantified by the sub-normalization factor \( \alpha \), the number of ancillary qubits, and the gate complexity required to implement \( U \).

General-purpose techniques for block encoding arbitrary matrices are well established, including QRAM-based constructions~\cite{qram}, linear-combination-of-unitaries (LCU)~\cite{lcu_childs} methods, Fast Approximate Block Encoding (FABLE)~\cite{camps2022fable} and the recent Binary Tree Block Encoding (BITBLE) technique~\cite{BITBLE}. Although these approaches are broadly applicable, they do not explicitly exploit the algebraic structure present in many physically relevant operators. As a result, when applied to structured matrices such as discretized differential operators, their resource requirements can be significantly larger than those achievable with structure-aware constructions~\cite{kuklinski2025efficientblockencodingsrequirestructure}. For sparse matrices, black-box sparse-access models~\cite{Gily_n_2019,low2019hamiltonian} provide an alternative abstraction for block encoding; however, these models do not directly translate into explicit, resource-efficient quantum circuits. As a result, the construction of explicit block encodings for structured sparse matrices has attracted significant attention. For example, in one dimension, explicit constructions have been developed for the discrete Laplacian with periodic and Dirichlet boundary conditions~\cite{camps_explicit_2023}, exploiting its tridiagonal and circulant structure. These ideas have been extended to higher dimensions in special cases, including explicit block encodings for the two-dimensional Dirichlet Laplacian~\cite{S_nderhauf_2024} and for \( D \)-dimensional Laplacians with periodic boundary conditions and uniform grid sizes~\cite{sturm2025efficientexplicitblockencoding}.

In particular, the latter work~\cite{sturm2025efficientexplicitblockencoding} has provided an explicit and resource-efficient block encoding of the \( D \)-dimensional discrete Laplacian with periodic boundary conditions on equal grids, achieving optimal sub-normalization and logarithmic gate complexity in the system size when \( D \) is a power of two. However, its scope is restricted to fully periodic boundary conditions and identical grid resolutions along each spatial dimension. From the perspective of both numerical modeling and physical applications, these restrictions are significant. Discrete Laplacians frequently arise with mixed boundary conditions, such as periodic boundaries along selected directions and Dirichlet or Neumann boundaries along others, as well as with unequal grid sizes reflecting anisotropic spatial resolution.

In this work, we present a generalized framework for block encoding discrete Laplacian operators that systematically overcomes these limitations. Building directly on the explicit periodic construction described above, we extend the methodology to support arbitrary spatial dimension, heterogeneous grid resolutions, and fully mixed boundary conditions within a single, modular circuit architecture. Our approach retains the exactness and favorable scaling of the periodic case while incorporating Dirichlet and Neumann boundary conditions through structured modifications of the encoding. The resulting block encodings are fully explicit, and achieve improved sub-normalization factors and lower number of quantum gates relative to existing approaches.
 
\medskip 
 
The contributions of this paper can be summarized as follows:
\begin{enumerate}
\item  We provide a unified block encoding construction for discrete Laplacians
with Dirichlet, periodic, and Neumann boundary conditions.

\item Our framework supports arbitrary spatial dimensions and allows different boundary conditions and grid sizes to be assigned independently along each coordinate axis, enabling mixed-boundary and anisotropic discretizations.

\item 
We provide both analytical gate-complexity estimates and practical circuit-level benchmarks for our block encoding constructions and existing explicit methods. For practical benchmarks, resource usage is evaluated via \texttt{Qiskit} transpilation for an IBM hardware, in terms of gate counts and depth, while success probabilities are computed via the \texttt{Aer} simulator for representative input states.

\end{enumerate}

The remainder of this paper is organized as follows. In Section~\ref{background} we introduce the discrete Laplacian operator and review the concept of block encoding, fixing notation and conventions used throughout the paper. Section~\ref{explicit_constructions} presents the generalized block encoding construction for Laplacians, starting from $1$ dimension and then generalizing this to arbitrary dimensions, mixed boundary conditions, and unequal grid sizes. Section~\ref{resource_estimation} first derives analytical gate-complexity estimates for our constructions and then presents a comparative study against existing methods based on circuits transpiled for \texttt{ibm\_torino}. Finally, Section~\ref{conclusion} summarizes our findings and discusses potential applications and extensions. All quantum circuits and benchmarking code used in this work are implemented in \texttt{Qiskit} \cite{javadiabhari2024quantumcomputingqiskit} and are made publicly available to facilitate reproducibility and further investigation. The corresponding implementation is available on this \href{https://github.com/TDC28/qamp2025}{\textcolor{blue}{GitHub repo}}~\cite{tdc28_qamp2025}.

\section{Discrete Laplacians and Block Encoding Preliminaries}
\label{background}

\subsection{Finite-Difference Laplacian Operators}
\label{subsec:laplacian}

We begin by introducing the Laplacian operator in its continuous form and then describe its discretization on a finite grid. The discrete Laplacian operators defined here constitute the central objects that are later block encoded into quantum circuits.

Let \( \Omega \subset \mathbb{R}^D \) be a rectangular domain, where \( D \) denotes the spatial dimension. The Laplacian acting on a sufficiently smooth function \( v : \Omega \to \mathbb{R} \) is defined as
\[
\Delta v(\mathbf{x}) = \sum_{d=0}^{D-1} \frac{\partial^2 v(\mathbf{x})}{\partial x_d^2},
\qquad
\mathbf{x} = (x_0, x_1, \dots, x_{D-1}).
\]
The operator is complemented by boundary conditions imposed on \( \partial \Omega \), such as periodic, Dirichlet, or Neumann boundary conditions. These boundary conditions determine how the operator behaves at the domain boundaries and directly influence the algebraic structure of its discrete counterpart.

To obtain a discrete representation, we introduce a uniform Cartesian grid along each coordinate direction. For the \( d \)-th dimension, let
\[
\Omega_d = \{ x^{(d)}_0 + j h_d \mid j = 0,1,\dots,N_d - 1 \}
\]
denote a discretization with grid spacing \( h_d > 0 \) and \( N_d = 2^{n_d} \) grid points.

The full \( D \)-dimensional grid is given by
\[
\Omega = \Omega_0 \times \Omega_1 \times \cdots \times \Omega_{D-1} \subset \mathbb{R}^D,
\]
with total number of grid points
\[
N := \prod_{d=0}^{D-1} N_d .
\]
Functions defined on \( \Omega \) are represented as vectors in \( \mathbb{R}^N \) by fixing an ordering of the grid points.

In one dimension, the second derivative is approximated using the standard second-order central finite-difference scheme,
\[
\frac{\partial^2 v}{\partial x^2}(x_j)
\;\approx\;
\frac{v_{j+1} - 2 v_j + v_{j-1}}{h^2},
\]
where \( h \) denotes the grid spacing. This approximation yields an \( N_d \times N_d \) discrete Laplacian matrix of the form
\[
L^{(N_d)}_{h_d} = \frac{1}{h_d^2}
\begin{bmatrix}
* & * & * & * & \cdots & * \\
1 & -2 & 1 & 0 & \cdots & 0 \\
0 & 1 & -2 & 1 & \cdots & 0 \\
\vdots & \vdots & \vdots & \ddots & \vdots & \vdots \\
0 & 0 & 0 & \cdots & -2 & 1 \\
* & * & * & * & \cdots & *
\end{bmatrix},
\]
where the interior stencil is identical for all boundary conditions, and the first and last rows depend on the imposed boundary conditions.

At the continuous level, the boundary conditions considered in this work are given by
\begin{align*}
\text{Periodic:} \quad & v(0) = v(L), \\
\text{Dirichlet:} \quad & v(0) = v(L) = 0, \\
\text{Neumann:} \quad & \frac{\partial v}{\partial x}(0) = \frac{\partial v}{\partial x}(L) = 0,
\end{align*}
where \( L \) denotes the length of the one-dimensional domain.

Upon discretization, these conditions determine how function values outside the computational grid are handled. Periodic boundary conditions identify values across opposite boundaries, leading to wrap-around couplings between the first and last grid points. Dirichlet boundary conditions fix boundary values and eliminate them from the system, while Neumann boundary conditions enforce vanishing derivatives by reflecting values at the boundary. As a result, only the first and last rows of the discrete Laplacian matrix differ between boundary conditions.

The corresponding one-dimensional discrete Laplacian matrices are given by
\[
L_{\mathrm{p}} = \frac{1}{h_d^2}
\begin{bmatrix}
-2 & 1 & 0 & \cdots & 0 & 1 \\
1 & -2 & 1 & \cdots & 0 & 0 \\
\vdots & \vdots & \vdots & \ddots & \vdots & \vdots \\
0 & 0 & 0 & \cdots & -2 & 1 \\
1 & 0 & 0 & \cdots & 1 & -2
\end{bmatrix},
\quad
L_{\mathrm{d}} = \frac{1}{h_d^2}
\begin{bmatrix}
-2 & 1 & 0 & \cdots & 0 & 0 \\
1 & -2 & 1 & \cdots & 0 & 0 \\
\vdots & \vdots & \vdots & \ddots & \vdots & \vdots \\
0 & 0 & 0 & \cdots & -2 & 1 \\
0 & 0 & 0 & \cdots & 1 & -2
\end{bmatrix},
\]
\[
L_{\mathrm{n}} = \frac{1}{h_d^2}
\begin{bmatrix}
-1 & 1 & 0 & \cdots & 0 & 0 \\
1 & -2 & 1 & \cdots & 0 & 0 \\
\vdots & \vdots & \vdots & \ddots & \vdots & \vdots \\
0 & 0 & 0 & \cdots & -2 & 1 \\
0 & 0 & 0 & \cdots & 1 & -1
\end{bmatrix}.
\]

The discrete Laplacian in \( D \) dimensions is obtained by summing one-dimensional Laplacians acting independently along each coordinate direction. Denoting by \( I^{(N_d)} \) the \( N_d \times N_d \) identity matrix, the full operator \( L \in \mathbb{R}^{N \times N} \) can be written as
\begin{equation}
L =
\sum_{d=0}^{D-1}
\left(
\bigotimes_{k=d+1}^{D-1} I^{(N_k)}
\right)
\otimes
L^{(N_d)}_{h_d}
\otimes
\left(
\bigotimes_{k=0}^{d-1} I^{(N_k)}
\right).
\label{laplace_kron}
\end{equation}
This Kronecker-sum structure reflects the separability of the Laplacian and will be central to the block encoding constructions presented later. Importantly, each one-dimensional operator \( L^{(N_d)}_{h_d} \) may correspond to a different boundary condition and grid spacing, allowing the discrete Laplacian to naturally incorporate mixed boundary conditions and anisotropic discretizations.

For block encoding, the operator must be scaled so that its spectral norm does not exceed unity. We therefore introduce scaled one-dimensional Laplacians. For each boundary condition \( b_d \in \{\mathrm{p},\mathrm{d},\mathrm{n}\} \), 
the spectrum of \( L^{(N_d)}_{h_d,b_d} \) is contained in 
$
\left[-\frac{4}{h_d^2},\,0\right],
$
so that
$
\|L^{(N_d)}_{h_d,b_d}\|_2 \le \frac{4}{h_d^2}.
$
Consequently, we define the scaled $1D$ Laplacian as:
\[
\widetilde{L}^{(N_d)}_{b_d}
:=
\frac{h_d^2}{4}
L^{(N_d)}_{h_d,b_d}.
\]
This scaling removes the grid-spacing dependence and ensures that $
\bigl\|
\widetilde{L}^{(N_d)}_{b_d}
\bigr\|_2
\le 1.$ For the multi-dimensional operator, the largest eigenvalue in magnitude satisfies
\[
\lambda_{\max}
\left(
L^{(D)}_{b_0,\dots,b_{D-1}}
\right)
=
4 \sum_{d=0}^{D-1} \frac{1}{h_d^2}.
\]
We therefore introduce the global scaling factor
\[
\Lambda
:=
\left(
\sum_{d=0}^{D-1} \frac{1}{h_d^2}
\right)^{-1},
\]
and define the scaled multi-dimensional Laplacian
\[
\widetilde{L}^{(D)}_{b_0,\dots,b_{D-1}}
:=
\frac{\Lambda}{4}
L^{(D)}_{b_0,\dots,b_{D-1}}.
\]

Substituting the Kronecker-sum form for $L^{(D)}_{b_0,\dots,b_{D-1}}$ from Eq.~\ref{laplace_kron} yields
\begin{equation}
\widetilde{L}^{(D)}_{b_0,\dots,b_{D-1}}
=
\sum_{d=0}^{D-1}
\left(
\bigotimes_{k=d+1}^{D-1} I^{(N_k)}
\right)
\otimes
\omega_d \,
\widetilde{L}^{(N_d)}_{b_d}
\otimes
\left(
\bigotimes_{k=0}^{d-1} I^{(N_k)}
\right),
\label{eq:ND-scaled}
\end{equation}
where the dimension-dependent weights are
\begin{equation}
\omega_d
:=
\frac{\Lambda}{h_d^2}
=
\frac{1/h_d^2}{\sum_{i=0}^{D-1} 1/h_i^2},
\qquad
\sum_{d=0}^{D-1} \omega_d = 1.
\label{omega_def}
\end{equation}

By construction,
$
\bigl\|
\widetilde{L}^{(D)}_{b_0,\dots,b_{D-1}}
\bigr\|_2
\le 1,
$
so that \( \widetilde{L}^{(D)}_{b_0,\dots,b_{D-1}} \) is directly suitable for block encoding.

\subsection{Block-Encoding}
\label{subsec:blockencoding}

Block encoding provides for a systematic way to represent a generally non-unitary matrix as a sub-block of a larger unitary operator that can be implemented as a quantum circuit.

Let \( A \in \mathbb{C}^{2^n \times 2^n} \) be a matrix acting on an \( n \)-qubit system register. Suppose there exist a scaling factor \( \alpha > 0 \), an error parameter \( \varepsilon \ge 0 \), and a unitary operator
\[
U_A \in \mathbb{C}^{2^{n+m} \times 2^{n+m}},
\]
acting on \( m \) ancilla qubits and \( n \) system qubits such that
\begin{equation}
\left\|
A
-
\alpha
\left(
\bra{0}^{\otimes m} \otimes I
\right)
U_A
\left(
\ket{0}^{\otimes m} \otimes I
\right)
\right\|_2
\le \varepsilon .
\label{eq:block encoding-def}
\end{equation}
Then \( U_A \) is called an \emph{\((\alpha,m,\varepsilon)\)-block encoding} of \( A \). In the special case \( \varepsilon = 0 \), the block encoding is said to be \emph{exact}, and \( U_A \) is referred to as an \emph{\((\alpha,m)\)-block encoding} of \( A \). Throughout this work, we focus exclusively on the construction of \emph{exact} block encodings.

$U_A$ may be equivalently expressed in block-matrix form as
\[
U_A =
\begin{pmatrix}
A / \alpha & * \\
* & *
\end{pmatrix},
\]
where the upper-left block is selected by projecting the ancilla register onto \( \ket{0}^{\otimes m} \), and the remaining blocks are irrelevant.

Operationally, a block encoding is implemented by a quantum circuit acting on \( m \) ancilla qubits initialized in the state \( \ket{0}^{\otimes m} \) and \( n \) system qubits initialized in an arbitrary input state \( \ket{\psi} \).

The action of the unitary \( U_A \) on the joint input state can be written as
\[
U_A
\left(
\ket{0}^{\otimes m} \otimes \ket{\psi}
\right)
=
\ket{0}^{\otimes m} \otimes \left( \frac{A}{\alpha} \ket{\psi} \right)
+
\sum_{s \neq 0} \ket{s} \otimes \ket{\phi_s},
\]
where \( \{ \ket{s} \} \) ranges over all computational basis states of the ancilla register except \( \ket{0}^{\otimes m} \), and \( \ket{\phi_s} \) are unnormalized states of the system register.

Upon measuring the ancilla qubits and postselecting on the outcome \( \ket{0}^{\otimes m} \), the system register is projected onto a state proportional to \( A \ket{\psi} \). The probability of this event is given by
\[
p_{\mathrm{succ}} = \frac{\| A \ket{\psi} \|^2}{\alpha^2},
\]
which we refer to as the \emph{block encoding success probability}. The sub-normalization factor \( \alpha \) directly controls the efficiency of postselection based procedures and thus plays a central role in the performance of algorithms built on block encodings.

\section{Generalized Block Encoding of Multi-Dimensional Laplacians}
\label{explicit_constructions}

In this section, we present quantum circuits to block encode Laplacian operators with periodic, Dirichlet, or  Neumann boundary conditions. We then generalize our approach to N-dimensional Laplacian operators. Proof of correctness for each case is also provided.

\subsection{Block encoding a $1D$ Laplacian}

\subsubsection{Periodic boundary conditions}

For block encoding the $D-$dimensional Laplacian with periodic boundary conditions across all dimensions, an efficient approach has already been suggested in~\cite{sturm2025efficientexplicitblockencoding}. For completion, we just state the theorem here and for complete proof refer the readers to~\cite{sturm2025efficientexplicitblockencoding}. 

First, we define the shift operators $S^-$ and $S^+$ acting on an $n$-qubit register encoding computational basis states
$\ket{j}$, $j=0,\dots,N-1$, as
\[
S^- \ket{j} = \ket{j-1 \!\!\!\!\pmod{N}}, 
\qquad
S^+ \ket{j} = \ket{j+1 \!\!\!\!\pmod{N}}.
\]
These operators implement cyclic left and right shifts of the discrete grid index.

\begin{theorem}[Block encoding of the 1D periodic Laplacian]
Let $U_p$ denote the unitary implemented by the quantum circuit shown in Fig.\ref{fig:1DP_circuit}, acting on $n$ system qubits and $m=2$ ancilla qubits. Then $U_p$ is an exact $(1,2,0)$-block encoding of the scaled one-dimensional periodic Laplacian 
$\widetilde{L}^{(1)}_{\mathrm{p}}$, i.e.,
\[
\left( \bra{0}^{\otimes 2} \otimes I \right)
U_p
\left( \ket{0}^{\otimes 2} \otimes I \right)
=
\widetilde{L}^{(1)}_{\mathrm{p}} .
\]

\begin{figure}[h]
\centering
\scalebox{1.4}{
\begin{quantikz}
    \lstick{$\ket{j}$} & \qwbundle{n} & & \gate{S^-} & \gate{S^+} & & & \\
    \lstick{$\ket{l_0}$} & \gate{H} & \gate{Z} & & \ctrl{-1} & \gate{H} & \meter{} & \\
    \lstick{$\ket{l_1}$} & \gate{H} & \gate{Z} & \ctrl[open]{-2} & & \gate{H} & \meter{} &
\end{quantikz}
}
\caption{Block encoding circuit for the 1D Laplacian with periodic boundary, as provided in~\cite{sturm2025efficientexplicitblockencoding}.}
\label{fig:1DP_circuit}
\end{figure}
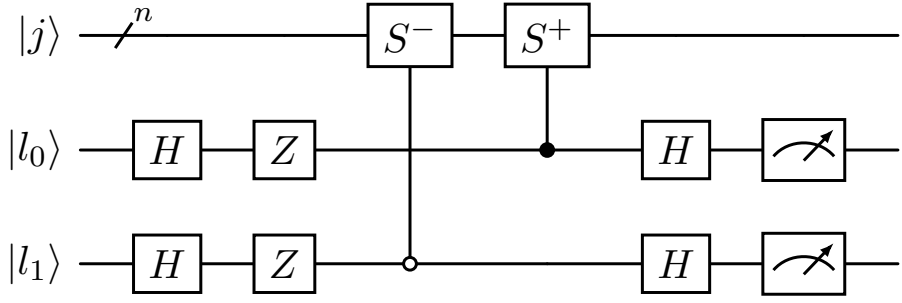

Equivalently, $U_p$ has the block form
\[
U_p
=
\begin{pmatrix}
\widetilde{L}^{(1)}_{\mathrm{p}} & * \\
* & *
\end{pmatrix}.
\]
Consequently, when the ancilla register is initialized in $\ket{0}^{\otimes 2}$ and postselected on the outcome $\ket{0}^{\otimes 2}$, the induced transformation on the system register is $\widetilde{L}^{(1)}_{\mathrm{p}}$.
\end{theorem}

\subsubsection{Dirichlet boundary conditions}

\begin{theorem}[Block encoding of the 1D Dirichlet Laplacian]
Let $U_{\mathrm{d}}$ denote the unitary implemented by the circuit shown in Fig.~\ref{fig:1DD_circuit},
acting on $n$ system qubits, two ancilla qubits $\ell_0,\ell_1$, and one
additional ancilla qubit $\mathrm{del}$. Then $U_{\mathrm{d}}$ is an exact
$(1,3,0)$-block encoding of the scaled one-dimensional Dirichlet Laplacian
$\widetilde{L}^{(1)}_{\mathrm{d}}$, i.e.,
\[
(\bra{0}^{\otimes 3} \otimes I)
U_{\mathrm{d}}
(\ket{0}^{\otimes 3} \otimes I)
=
\widetilde{L}^{(1)}_{\mathrm{d}} .
\]

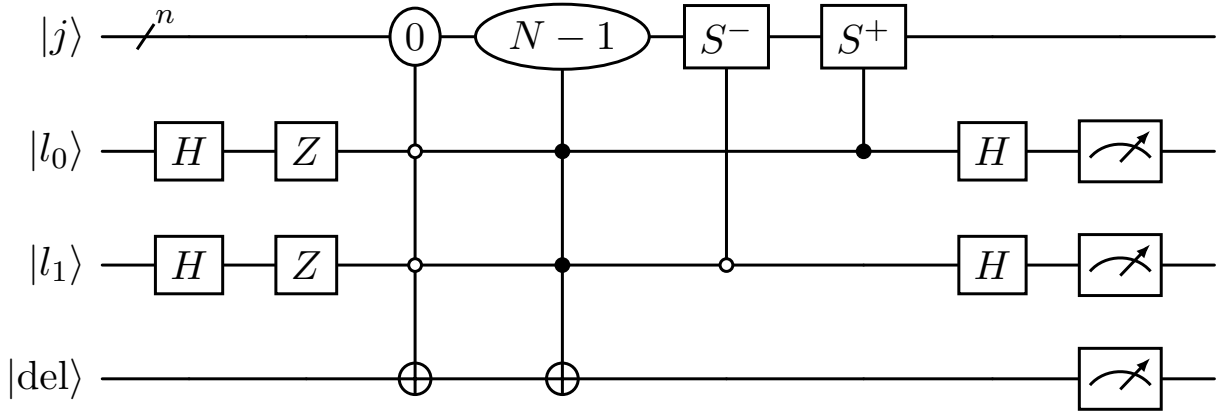
\begin{figure}[h]
\centering
\scalebox{1.4}{
\begin{quantikz}
\lstick{$\ket{j}$} & \qwbundle{n} & & \gate[style={shape=ellipse, fill=white, inner sep=-2pt}]{0}\wire[d]{q} & \gate[style={shape=ellipse, fill=white, inner sep=-2pt}]{N-1}\wire[d]{q} &\gate{S^-}&\gate{S^+} & & & \\
\lstick{$\ket{l_0}$} & \gate{H} & \gate{Z} & \ctrl[open]{1} & \ctrl{1} & & \ctrl{-1} & \gate{H} & \meter{} & \\
\lstick{$\ket{l_1}$} & \gate{H} & \gate{Z} & \ctrl[open]{1} & \ctrl{1} & \ctrl[open]{-2} && \gate{H} & \meter{} & \\
\lstick{$\ket{\text{del}}$} & & & \targ{} & \targ{} & & & & \meter{} &
\end{quantikz}
}
\caption{Block encoding circuit for the 1D Laplacian with Dirichlet boundary.}
\label{fig:1DD_circuit}
\end{figure}

\end{theorem}

\begin{proof}

We verify the block encoding relation on computational basis states
$\ket{j}$, $j=0,\dots,N-1$.

For interior points $1 \le j \le N-2$, the additional boundary-controlled
operations are inactive. The circuit therefore reduces exactly to the
periodic construction shown previously , and we obtain
\[
\ket{0}^{\otimes 3}\ket{j}
\longmapsto
\ket{0}^{\otimes 3}
\widetilde{L}^{(1)}_{\mathrm{d}}\ket{j}
+
\sum_{s\neq 0}
\ket{s}\ket{\phi_s(j)},
\]
where
\[
\widetilde{L}^{(1)}_{\mathrm{d}}\ket{j}
=
\frac14
\left(
\ket{j-1} - 2\ket{j} + \ket{j+1}
\right).
\]

It therefore remains to verify the boundary cases $j=0$ and $j=N-1$.

\medskip
\noindent
\noindent
\textbf{Case $j=0$:}

We adopt the register ordering
$\ket{\mathrm{del}}\ket{\ell}\ket{j}$. After Hadamard gates on the $\ell$ register and application of $Z$ gates, the joint state is
\[
\frac12
\ket{0} \left(
\ket{00} - \ket{01} - \ket{10} + \ket{11}
\right)
\ket{0}.
\]

The boundary-controlled operation flips the $\mathrm{del}$ ancilla precisely for the component corresponding to $j=0$ and $\ell=00$, yielding
\[
\frac12
\left(
\ket{1}\ket{00}
-\ket{0}\ket{01}
-\ket{0}\ket{10}
+\ket{0}\ket{11}
\right)
\ket{0}.
\]

After the controlled shift operations $S^\pm$, we obtain
\[
\frac12
\left(
\ket{1}\ket{00}\ket{N-1}
- \ket{0}\ket{01}\ket{0}
- \ket{0}\ket{10}\ket{0}
+ \ket{0}\ket{11}\ket{1}
\right).
\]

Finally, after the final Hadamards on the $\ell$ register, the component associated with
$\ket{0}^{\otimes 3}$ in the ancilla registers is
\[
\frac14
\ket{0}^{\otimes 3}
\left(
-2\ket{0} + \ket{1}
\right).
\]

All remaining components are orthogonal to $\ket{0}^{\otimes 3}$
in the ancilla space. Therefore,
\[
(\bra{0}^{\otimes 3}\otimes I)
U_{\mathrm{d}}
(\ket{0}^{\otimes 3}\otimes \ket{0})
=
\frac14
\left(
-2\ket{0} + \ket{1}
\right)
= \widetilde{L}^{(1)}_{\mathrm{d}}\ket{0}.
\]

\medskip
\noindent
\textbf{Case $j=N-1$:}

An analogous computation shows that the boundary controlled operation now
activates for the component with $\ell=11$, leading after the shift
operations and uncomputation to
\[
(\bra{0}^{\otimes 3}\otimes I)
U_{\mathrm{d}}
(\ket{0}^{\otimes 3}\otimes \ket{N-1})
=
\frac14
\left(
\ket{N-2} - 2\ket{N-1}
\right)
= \widetilde{L}^{(1)}_{\mathrm{d}}\ket{N-1}.
\]
\medskip
Since the action agrees with
$\widetilde{L}^{(1)}_{\mathrm{d}}$ on all basis vectors and the
construction is exact, the stated block encoding relation follows.
\end{proof}

\subsubsection{ Neumann boundary conditions}

\begin{theorem}[Block encoding of the 1D Neumann Laplacian]
Let $U_{\mathrm{n}}$ denote the unitary implemented by the circuit shown in Fig.~\ref{fig:1DV_circuit},
acting on $n$ system qubits, two ancilla qubits $\ell_0,\ell_1$, and one
additional ancilla qubit $\mathrm{del}$. Then $U_{\mathrm{n}}$ is an exact
$(1,3,0)$-block encoding of the scaled one-dimensional Neumann Laplacian
$\widetilde{L}^{(1)}_{\mathrm{n}}$, i.e.,
\[
(\bra{0}^{\otimes 3} \otimes I)
U_{\mathrm{n}}
(\ket{0}^{\otimes 3} \otimes I)
=
\widetilde{L}^{(1)}_{\mathrm{n}} .
\]

\begin{figure}[h]
\centering
\scalebox{1.2}{
\begin{quantikz}
\lstick{$\ket{j}$} & \qwbundle{n} & & \gate[style={shape=ellipse, fill=white, inner sep=-2pt}]{0}\wire[d]{q} & \gate[style={shape=ellipse, fill=white, inner sep=-2pt}]{0}\wire[d]{q} & \gate[style={shape=ellipse, fill=white, inner sep=-2pt}]{N-1}\wire[d]{q} & \gate[style={shape=ellipse, fill=white, inner sep=-2pt}]{N-1}\wire[d]{q} &\gate{S^-}&\gate{S^+} & & & \\
\lstick{$\ket{l_0}$} & \gate{H} & \gate{Z} & \ctrl[open]{1} & \ctrl{1} & \ctrl{1} & \ctrl{1} & & \ctrl{-1} & \gate{H} & \meter{} & \\
\lstick{$\ket{l_1}$} & \gate{H} & \gate{Z} & \ctrl[open]{1} & \ctrl[open]{1} & \ctrl{1} & \ctrl[open]{1} & \ctrl[open]{-2} && \gate{H} & \meter{} & \\
\lstick{$\ket{\text{del}}$} & & & \targ{} & \targ{} & \targ{} & \targ{} & & & & \meter{} &
\end{quantikz}
}
\caption{Block encoding circuit for the 1D Laplacian with  Neumann boundary.}
\label{fig:1DV_circuit}
\end{figure}
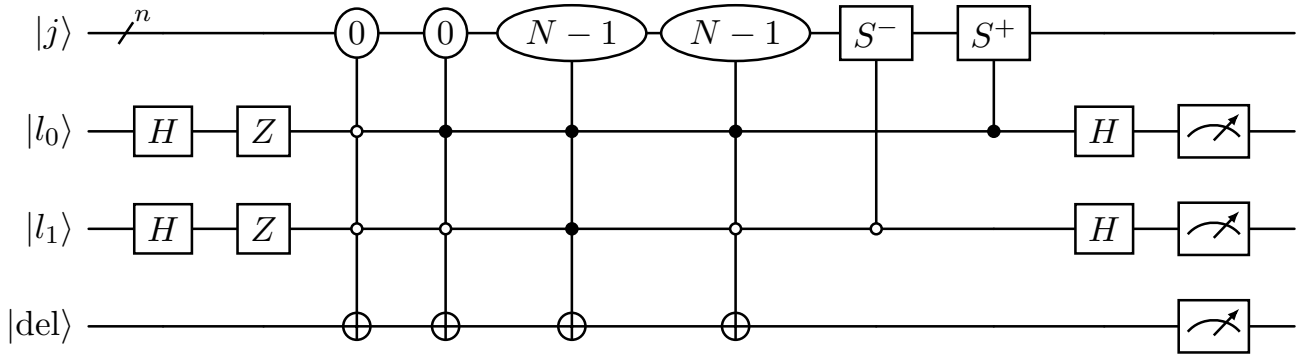

\end{theorem}

\begin{proof}
We verify the block encoding relation on computational basis states
$\ket{j}$, $j=0,\dots,N-1$. 

\medskip

For $1 \le j \le N-2$, none of the boundary comparator circuits are activated.  The construction therefore reduces to the periodic interior stencil,
\[
\widetilde{L}^{(1)}_{\mathrm{n}}\ket{j}
=
\frac14\bigl(\ket{j-1}-2\ket{j}+\ket{j+1}\bigr),
\]
and the block encoding relation follows exactly as in the periodic case.

\medskip
\noindent
\textbf{Case $j=0$:}

After Hadamard gates on the $\ell$ register and application of $Z$ gates, the joint state is
\[
\frac12
\ket{0} \bigl(
\ket{00}-\ket{01}-\ket{10}+\ket{11}
\bigr)
\ket{0}.
\]

Applying the controlled shifts gives
\[
\frac12
\bigl(
\ket{1}\ket{00}\ket{N-1}
-\ket{1}\ket{01}\ket{0}
-\ket{0}\ket{10}\ket{0}
+\ket{0}\ket{11}\ket{1}
\bigr).
\]

Uncomputing the $\ell$ register and extracting only the component
associated with $\ket{0}^{\otimes 3}$ in the ancilla space yields
\[
\frac14
\ket{0}^{\otimes 3}
\bigl(
-\ket{0}+\ket{1}
\bigr).
\]

All remaining components lie in subspaces orthogonal to
$\ket{0}^{\otimes 3}$.  Hence
\[
(\bra{0}^{\otimes 3}\otimes I)
U_{\mathrm{n}}
(\ket{0}^{\otimes 3}\otimes\ket{0})
=
\frac14
\bigl(
-\ket{0}+\ket{1}
\bigr)
= \widetilde{L}^{(1)}_{\mathrm{n}}\ket{0}.
\]

\medskip
\noindent
\textbf{Case $j=N-1$:}

An analogous calculation shows that for $j=N-1$ the last two boundary
comparators are activated, and projection onto
$\ket{0}^{\otimes 3}$ yields
\[
(\bra{0}^{\otimes 3}\otimes I)
U_{\mathrm{n}}
(\ket{0}^{\otimes 3}\otimes\ket{N-1})
=
\frac14
\bigl(
\ket{N-2}-\ket{N-1}
\bigr)
=\widetilde{L}^{(1)}_{\mathrm{n}}\ket{N-1}.
\]
\medskip
Since the action agrees with
$\widetilde{L}^{(1)}_{\mathrm{n}}$
on all computational basis states and the construction is exact,
the block encoding relation follows.
\end{proof}

\medskip
\noindent
\textbf{Remark (Success probability):}
All three constructions (periodic, Dirichlet, and Neumann) are exact
$(1,m,0)$-block encodings of the corresponding scaled Laplacian
$\widetilde{L}^{(1)}_{b}$ with the same scaling
\(
\widetilde{L}^{(1)}_{b} = \frac{h^2}{4} L^{(1)}_{h,b}.
\)
Consequently, for any normalized state $\ket{v}$, the success probability of postselecting the
ancilla register in $\ket{0}^{\otimes m}$ is
\[
p_{\mathrm{succ}}
=
\|\widetilde{L}^{(1)}_{b} \ket{v}\|_2^2 .
\]

This scales as
\[
p_{\mathrm{succ}}
\sim
\frac{h^4}{16}\|L^{(1)}_b \ket{v}\|_2^2,
\]
independently of the chosen boundary condition.

\subsection{Block encoding N-dimensional Laplacian operators}

Recall that the scaled D-dimensional Laplacian operator is given by 
$$\widetilde{L}^{(D)}_{b_0,\dots,b_{D-1}}
=
\sum_{d=0}^{D-1}
\left(
\bigotimes_{k=d+1}^{D-1} I^{(N_k)}
\right)
\otimes
\omega_d \,
\widetilde{L}^{(N_d)}_{b_d}
\otimes
\left(
\bigotimes_{k=0}^{d-1} I^{(N_k)}
\right).$$
Extending the 1-dimensional approaches presented above to this case is easy due to the Kronecker-sum structure of $\widetilde{L}^{(D)}$. We take advantage of that structure to block encode N-dimensional Laplacian operators by adding a new register containing $\lceil \log D\rceil$ qubits, which we call the $\ket{k}$ register. On this register, we prepare the state
\begin{equation}
U_{prep\_k}\ket{k} \;=\; \sum_{d=0}^{D-1}\sqrt{\omega_d}\,\ket{d},
\end{equation}
where $\omega_d$ has been defined in Eq.~\ref{omega_def} and depends on grid spacing along each dimension. 

We add $\ket{j}$ registers, indexed as $\ket{j^{(d)}}$, so as to have one for each of the $D$ dimensions. The registers in the quantum circuit are then given and ordered by $\ket{k} \ket{\text{del}} \ket{l} \ket{j^{(D-1)}}\ket{j^{(D-2)}} \dots\ket{j^0}$. \\

\begin{theorem}[Block encoding of $D-$dimensional Laplacian] Let $U^{\mathrm{D}}$ denote the unitary implemented by the circuit shown in Fig.~\ref{fig:DD_circuit},
acting on $n$ system qubits, $3+\lceil \log D\rceil$ ancilla qubits $\ell_0,\ell_1$. Then $U^{\mathrm{D}}$ is an exact
$(1,3+\lceil \log D\rceil,0)$-block encoding of the scaled $D-$dimensional Laplacian
$\widetilde{L}^{(D)}_{b_0,\dots,b_{D-1}}$, i.e.,
\[
(\bra{0}^{\otimes (3+\lceil \log D\rceil)} \otimes I)
U^{\mathrm{D}}
(\ket{0}^{\otimes (3+\lceil \log D\rceil)} \otimes I)
=
\widetilde{L}^{(D)}_{b_1,\dots,b_D} .
\]

\begin{figure}[h]
\centering
\scalebox{0.75}{
\begin{quantikz}
    \lstick{$\ket{j^{(0)}}$} &&&\gate[wires=8]{U_{b_0}}&\gate{S^-}&\gate{S^+}&&&&&&&&&\\
    \lstick{$\ket{j^{(1)}}$} &&&&& &\gate[wires=7]{U_{b_1}}& \gate{S^-}&\gate{S^+}&&&&&&\\
    \lstick \vdots \\
    \lstick \vdots \\
    \lstick{$\ket{j^{(D-1)}}$} &&&&&&&&& &\gate[wires=4]{U_{b_{D-1}}} & \gate{S^-}&\gate{S^+}&& \\ \lstick{$\ket{l_0}$}&\gate{H}&\gate{Z}&&&\ctrl{-5}\wire[d]{q}&&&\ctrl{-4}\wire[d]{q} &\cdots && &\ctrl{-1} &\gate{H}&\meter{}\\
    \lstick{$\ket{l_1}$}&\gate{H}&\gate{Z}&&\ctrl[open]{-6}\wire[d]{q}&\wire[d]{q}&&\ctrl[open]{-5}\wire[d]{q}
    &\wire[d]{q}&\cdots&&\ctrl[open]{-2}\wire[d]{q}&\wire[u]{q}\wire[d]{q}&\gate{H}&\meter{}\\
    \lstick{$\ket{\text{del}}$}&&&&\wire[d]{q}&\wire[d]{q}&&\wire[d]{q}&\wire[d]{q}&\cdots&\wire[d]{q}&&&&\meter{}\\
    \lstick{$\ket{k}$}&\gate{U_{prep\_k}}&&\gate[style={shape=ellipse, fill=white, inner sep=-2pt}]{0}\wire[u]{q}&\gate[style={shape=ellipse, fill=white, inner sep=-2pt}]{0}&\gate[style={shape=ellipse, fill=white, inner sep=-2pt}]{0}&\gate[style={shape=ellipse, fill=white, inner sep=-2pt}]{1}\wire[u]{q}&\gate[style={shape=ellipse, fill=white, inner sep=-2pt}]{1}&\gate[style={shape=ellipse, fill=white, inner sep=-2pt}]{1} &\cdots &\gate[style={shape=ellipse, fill=white, inner sep=-2pt}]{D-1}\wire[u]{q}&\gate[style={shape=ellipse, fill=white, inner sep=-2pt}]{D-1}\wire[u]{q}&\gate[style={shape=ellipse, fill=white, inner sep=-2pt}]{D-1}\wire[u]{q}&\gate{U_{prep\_k}^\dagger}&\meter{}\\
\end{quantikz}
}
\caption{Block encoding circuit for the $D-$dimensional Laplacian with boundaries $b_0, b_1, \cdots b_{D-1}$ where each $b_i$ can be Dirichlet, periodic or  Neumann. The circuits for each $U_{b_i}$ can be inferred from the 1D constructions provided earlier. For example, if a dimension $m$ has a periodic boundary, $U_{b_m}$ is just the Identity matrix.}
\label{fig:DD_circuit}
\end{figure}
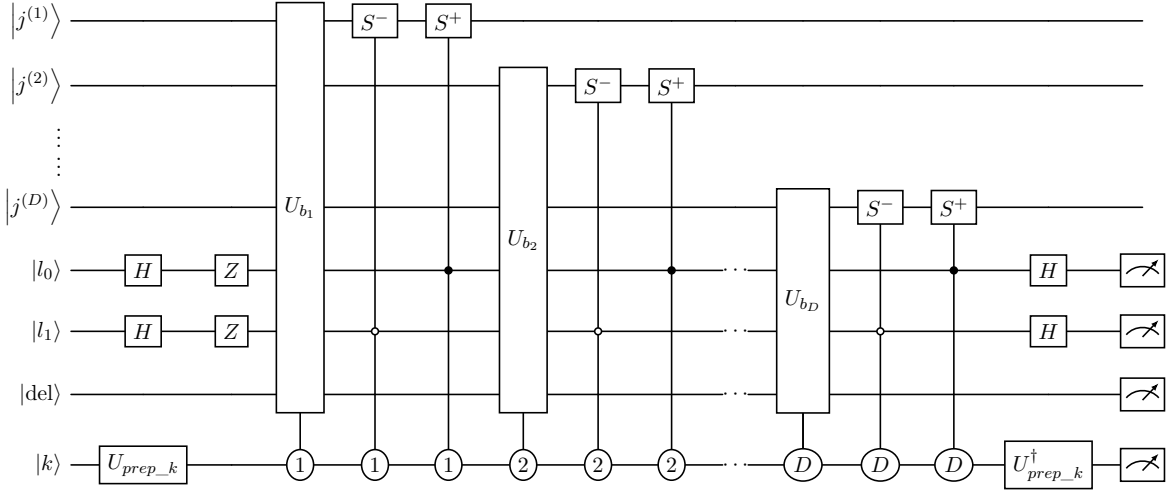
\end{theorem}

\begin{proof}
    Similar to the 1D cases, we verify the block encoding relation on computational basis states $\ket{j^{(d)}}$. After the application on $U_{prep\_k}$ on the $\ket{k}$ register, $H$ and $   Z$ gates on $\ket{l}=\ket{l_1l_0}$ registers, the joint state is
    $$\sum_{l=0}^{3}\sum_{d=0}^{D-1} (-1)^{l_0 + l_1} \sqrt{\omega_d} \ket{d}\ket{0}\ket{l}\ket{j^{(D-1)}}\ket{j^{(D-2)}}...\ket{j^{(0)}}.$$ \\

    The $U_{b_d}$ and $S^{+/-}$ operators controlled to state $d=0, 1, \dots, D-1$ iteratively apply the 1-dimensional scaled Laplacian operator with the boundary condition corresponding to the choice for dimension $d$. After their application, the joint state is
\begin{align*}
        \sqrt{\omega_{0}}\ket{0} \ket{0} \ket{0} \ket{j^{(D-1)}} &\dots \ket{j^{(1)}} \tilde{L}^{(N_{0})}_{b_{0}} \ket{j^{(0)}} \\
        + \sqrt{\omega_{1}} \ket{1} \ket{0} \ket{0} \ket{j^{(D-1)}} &\dots \tilde{L}^{(N_{1})}_{b_{1}} \ket{j^{(1)}} \ket{j^{(0)}} \\
        &\vdots \hspace{3.5cm} \\
        + \sqrt{\omega_{D-1}} \ket{D-1} \ket{0} \ket{0} \tilde{L}^{(N_{D-1})}_{b_{D-1}} &\ket{j^{(D-1)}} \dots \ket{j^{(1)}} \ket{j^{(0)}} \\
        + \sum_{d=0}^{D-1} \alpha_{d10}\ket{d} \ket{1} \ket{0} \ket{\psi^{(D-1)}_{10}}...\ket{\psi^{(0)}_{10}} + \sum_{d=0}^{D-1}&\sum_{r=0}^1\sum_{l=1}^3 \alpha_{drl}\ket{d}\ket{r}\ket{l}\ket{\psi^{(D-1)}_{rl}}...\ket{\psi^{(0)}_{rl}}.
    \end{align*}

    Finally, $U_{prep\_k}^\dagger$ is applied to the $\ket{k}$ register, where the joint state becomes

    \begin{align*}
        \omega_{0} \ket{0} \ket{0} \ket{0} \ket{j^{(D-1)}} & \dots \ket{j^{(1)}} \tilde{L}^{(N_{0})}_{b_{0}} \ket{j^{(0)}} \\
        + \omega_{1} \ket{0} \ket{0} \ket{0} \ket{j^{(D-1)}} & \dots \tilde{L}^{(N_{1})}_{b_{1}} \ket{j^{(1)}} \ket{j^{(0)}}\\
        &\vdots \hspace{3.5cm}\\
        + \omega_{D-1} \ket{0} \ket{0} \ket{0} \tilde{L}^{(N_{D-1})}_{b_{D-1}} &\ket{j^{(D-1)}} \dots \ket{j^{(1)}} \ket{j^{(0)}} \\
        + \sum_{d=0}^{D-1
        } \beta_{d10}\ket{d} \ket{1} \ket{0} \ket{\psi^{(D-1)}_{10}}...\ket{\psi^{(0)}_{10}}
        + &\sum_{d=0}^{D-1}\sum_{r=0}^1\sum_{l=1}^3 \beta_{drl}\ket{d}\ket{r}\ket{l}\ket{\psi^{(D-1)}_{rl}}...\ket{\psi^{(0)}_{rl}}.
    \end{align*}
    \begin{align*}
        = \ket{0} \ket{0} \ket{0} & \left( \sum_{d=0}^{D-1}
        \left[ \bigotimes_{k=d+1}^{D-1} I^{(N_k)} \right]
        \otimes \omega_d \tilde{L}_{b_d}^{(N_d)} \otimes
        \left[ \bigotimes_{k=0}^{d-1} I^{(N_k)} \right] \right)
        \left(\ket{j^{(D-1)}} \dots \ket{j^{(1)}} \ket{j^{(0)}} \right)  \\
        &+ \sum_{d=0}^{D-1} \beta_{d10}\ket{d} \ket{1} \ket{0} \ket{\psi^{(D-1)}_{10}}...\ket{\psi^{(0)}_{10}} + \sum_{d=0}^{D-1}\sum_{r=0}^1\sum_{l=1}^3 \beta_{drl}\ket{d}\ket{r}\ket{l}\ket{\psi^{(D-1)}_{rl}}...\ket{\psi^{(0)}_{rl}}.
    \end{align*}

    After a $\ket{0} \ket{0} \ket{0}$ measurement, the state in the $j$ registers collapses to
    $$\left(\sum_{d=0}^{D-1}
        \left[ \bigotimes_{k=d+1}^{D-1} I^{(N_k)} \right]
        \otimes \omega_d \tilde{L}_{b_d}^{(N_d)} \otimes
        \left[ \bigotimes_{k=0}^{d} I^{(N_k)} \right]\right) \left(\ket{j^{(D-1)}} \dots \ket{j^{(1)}} \ket{j^{(0)}} \right). $$
\end{proof}

For realistic systems, the $\ket{k}$ register will contain at most $2$ qubits. The below implementation in Fig. \ref{uprep_circuit} for $U_{prep\_k}$ can be used which is optimal~\cite{grover2002creatingsuperpositionscorrespondefficiently}. For the special cases $D=1$ and $D=2$, the $\ket{k}$ register contains only a single qubit and therefore only two amplitudes need to be prepared. In this case, the state preparation reduces to only the $R_y$.

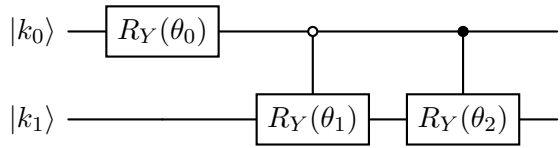
\begin{figure}[htbp]
\centering
\begin{quantikz}
\lstick{$\ket{k_0}$} & \gate{R_Y(\theta_0)} & \ctrl[open]{1}& \ctrl{1} & \\
\lstick{$\ket{k_1}$} & & \gate{R_Y(\theta_1)} & \gate{R_Y(\theta_2)} &
\end{quantikz}
\caption{State preparation circuit ($U_{prep\_k}$) for $D=3,4$. The rotation angles for the RY and controlled-RY gates are $\theta_0= 2\arccos{(\sqrt{\omega_0 + \omega_2})}$,
$\theta_1=2\arccos{(\sqrt\frac{\omega_0}{\omega_0 + \omega_2})}$,
$\theta_2=2\arccos{(\sqrt\frac{\omega_1}{\omega_1 + \omega_3})}$.}
\label{uprep_circuit}
\end{figure}

\section{Resource Estimates and Comparisons}
\label{resource_estimation}

\subsection{Theoretical resource estimates}
We briefly analyze the gate complexity of the proposed block encoding circuit, starting with the
$1D$ Laplacian with Neumann boundary conditions acting on an
$N=2^n$ dimensional Hilbert space.  The system register consists of
$n$ qubits and the construction uses three additional ancilla qubits. 
Our goal is to express the implementation cost in the Clifford+$T$ basis,
which is the standard metric for fault-tolerant quantum computation.

The circuit contains three types of components:
(i) single-qubit Clifford gates that act on ancillae registers - which doesn't scale with $n$,
(ii) a constant number of multi-controlled NOT operations whose number of
controls scales with $n$, and
(iii) controlled arithmetic corresponding to increment and decrement by one
(modulo $2^n$) on the $j$ register.

We focus on the second and third components, and explicitly count only the T-gates involved as they dominate the cost in fault tolerant quantum computing. For (ii), we note that an $m$-controlled $X$ can be implemented using $2m-3$ Toffoli gates given two
clean ancillas~\cite{KhattarGidney2025}. Each Toffoli can be compiled into Clifford+$T$ using the standard 7-$T$
decomposition~\cite{NielsenChuang2002}.
Hence, the $T$ count per MCX is $7(2m-3).$ Substituting $m=n+2$ and since four such MCX gates are required, the total $T$ count from the MCX gates become $56n + 28$. We note that different synthesis strategies may modify the constant prefactor multiplying $n$~\cite{barenco_1995, mcu_2025, selinger_2015, Gidney2015ControlledNots, rel_phase_toffoli}. However, all known constructions preserve linear scaling, which is the only property required for the complexity claims in this work.


In (iii),  conditioned on the ancilla qubits, the circuit performs addition or subtraction of one modulo $2^n$ on the $j$ register. Quantum addition circuits are known to admit linear Clifford+$T$ cost in the register size. Introducing an additional control qubit does not change the asymptotic behavior. Implementing an incrementer ($+1\mod n$) can be done using $3n$ Toffolis using atmost $5$ ancillas~\cite{KhattarGidney2025}. Compiling each Toffoli into 7 $T$ gates gives and doing this for both the increment and decrement operations, gives the total $T$ count as $42n$.  As above, constant factors depend on the particular adder design and the Toffoli synthesis technique employed~\cite{adders_1996, gidney2018halving, draper2004}.

Summing all non-Clifford contributions and recalling that $N = 2^n$, the total $T-$count scales as,
\begin{equation}
\Theta(\log N).
\end{equation}

\medskip

We now consider the $D$-dimensional Laplacian with Neumann boundaries along all the dimensions, under the assumption of equal
grid spacing across dimensions.
The system register decomposes into $D$ coordinate registers
$j^{(r)}$ of sizes $n_r$ qubits.
Let
\(
n_{\mathrm{tot}} = \sum_{r=1}^{D} n_r
\),
so that the matrix dimension is $N = 2^{n_{\mathrm{tot}}}$.

A selector register $k$ of size
\(
d = \lceil \log_2 D \rceil
\)
is prepared in a uniform superposition using Hadamards, and uncomputed at
the end of the procedure.
These are Clifford operations and therefore contribute no $T$ cost.

For each dimension $r$, the same boundary checking logic and controlled increment and decrement used in the one-dimensional construction is applied to
$j^{(r)}$, conditioned on the predicate $k=r$.
Implementing this condition adds $d$ additional controls to every
multi-controlled operation.
The necessary bit flips to convert the equality test into the necessary controls are Clifford and do not affect the $T$ count.
Each MCX now has $m' = n_r + 2 + d$ controls. The number of Toffolis is now $2m'-3$~\cite{KhattarGidney2025}. Compiling each Toffoli into 7 $T$ gates~\cite{NielsenChuang2002} and since four such $MCX$ operations are required for every $r$, the total $T$ count is found to be $56 n_r + 56 d + 28.$

As before, the incrementer requires $3n_r$ Toffolis. Each of these are controlled by $(d+1)$ control qubits. Since $d+3$ controlled $X$ can be implemented using $2(d+3)-3$ Toffolis, one can deduce that for the controlled incrementer, the Toffoli count scales as $O(n_rd)$. Accounting for both increment and decrement the total $T$ gate count still scales as $O(n_rd)$.
Summing over all dimensions and using $n_{\mathrm{tot}} = \sum_r n_r$, the total $T-$count scales as,
\begin{equation}
O(n_{\mathrm{tot}} d) = O(\log N \log D).
\end{equation}

As in the one-dimensional case, alternative synthesis choices may change constant prefactors. Therefore, the fault-tolerant cost of implementing the block encoding grows
only logarithmically both in the matrix size and the number of dimensions. This establishes the efficiency of the construction from the perspective of logical resource requirements. 

The Dirichlet and periodic variants use only a subset of the operations required in the Neumann construction. Consequently, their Clifford+$T$ cost is upper bounded by the result derived above, and the linear scaling in $n$ remains unchanged. For unequal grid sizes, the state preparation can instead be implemented using the circuit shown in Fig.~\ref{uprep_circuit}, which generalizes to arbitrary $D$, where $d=\lceil \log_2 D \rceil$. For $D \le 4$, which is sufficient for most systems, the $T$-count of this construction can be bounded by a constant independent of $n$. For general $D$, however, it is not known whether such state preparation can be implemented using $O(d)$ gates, although techniques such as in Ref.~\cite{Low_2024} seems promising.

\subsection{Comparative circuit implementation benchmarks}

To complement the analytical gate complexity estimates presented above, we perform circuit-level benchmarks of the proposed block encoding constructions and compare them with several existing explicit implementations, including the approaches in Camps \textit{et al.}~\cite{camps_explicit_2023}, Sünderhauf \textit{et al.}~\cite{S_nderhauf_2024}, FABLE~\cite{camps2022fable}, and BITBLE~\cite{BITBLE} where applicable.

All circuits are generated using \texttt{Qiskit} and transpiled to the native gate set of an IBM quantum processor- \texttt{ibm\_torino}. Hardware resource requirements are evaluated in terms of circuit depth, total gate count, and the number of two-qubit gates after transpilation. In addition, we evaluate the success probability of each block encoding construction using the \texttt{Aer} statevector simulator.

For the success-probability experiments, we consider representative normalized input vectors obtained by discretizing smooth test functions. Specifically, we use the functions
\begin{equation}
v_{1\mathrm{D}}(x) = \sin(2\pi x), \qquad
v_{2\mathrm{D}}(x,y) = \sin\!\big(2\pi(x+y)\big), \qquad
v_{3\mathrm{D}}(x,y,z) = \sin\!\big(2\pi(x+y+z)\big),
\label{eq:input_states}
\end{equation}
which are evaluated on the corresponding computational grids and normalized prior to state preparation. The benchmarks are performed for discrete Laplacian operators in one, two, and three spatial dimensions under various boundary conditions. For completeness, matrices encoded by the proposed circuits obtained by extracting the top-left block of the implemented block encoding unitary, is shown in Appendix~\ref{app:matrix_structure}.

\subsubsection{$1$D Laplacians}

For Dirichlet and periodic boundary conditions shown in Fig.~\ref{fig:1D_C_D} and Fig.~\ref{fig:1D_C_P} respectively, the proposed construction exhibits resource scaling comparable to that of~\cite{camps_explicit_2023}. In contrast, the explicit constructions produced by~\cite{camps2022fable}, and~\cite{BITBLE} exhibit significantly steeper scaling in all resource metrics, resulting in substantially larger circuits as the problem size increases. For Neumann boundary conditions shown in Fig.~\ref{fig:1D_FB_N}, the implementations in~\cite{camps2022fable}, and~\cite{BITBLE} require substantially larger circuits at increasing matrix dimensions.

Fig.~\ref{fig:SP_1D} shows the success probability of the different block encoding constructions evaluated using representative input states from~\eqref{eq:input_states}. While the success probability decreases with system size for all methods, the construction proposed in this work consistently achieves higher success probabilities across the tested problem sizes.

\begin{figure}[htbp]
    \centering
    \includegraphics[width=1\linewidth]{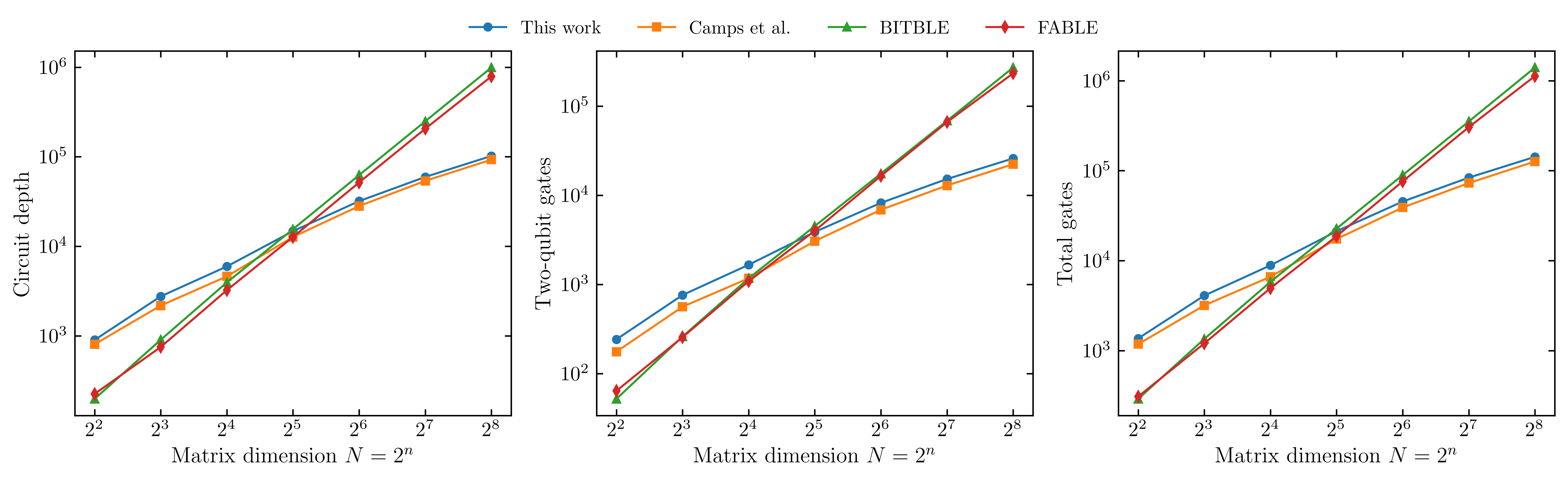}
    \caption{Comparison of resources for block encoding constructions of the $1$D discrete Laplacian with Dirichlet boundary conditions}
    \label{fig:1D_C_D}
\end{figure}

\begin{figure}[htbp]
    \centering
    \includegraphics[width=1\linewidth]{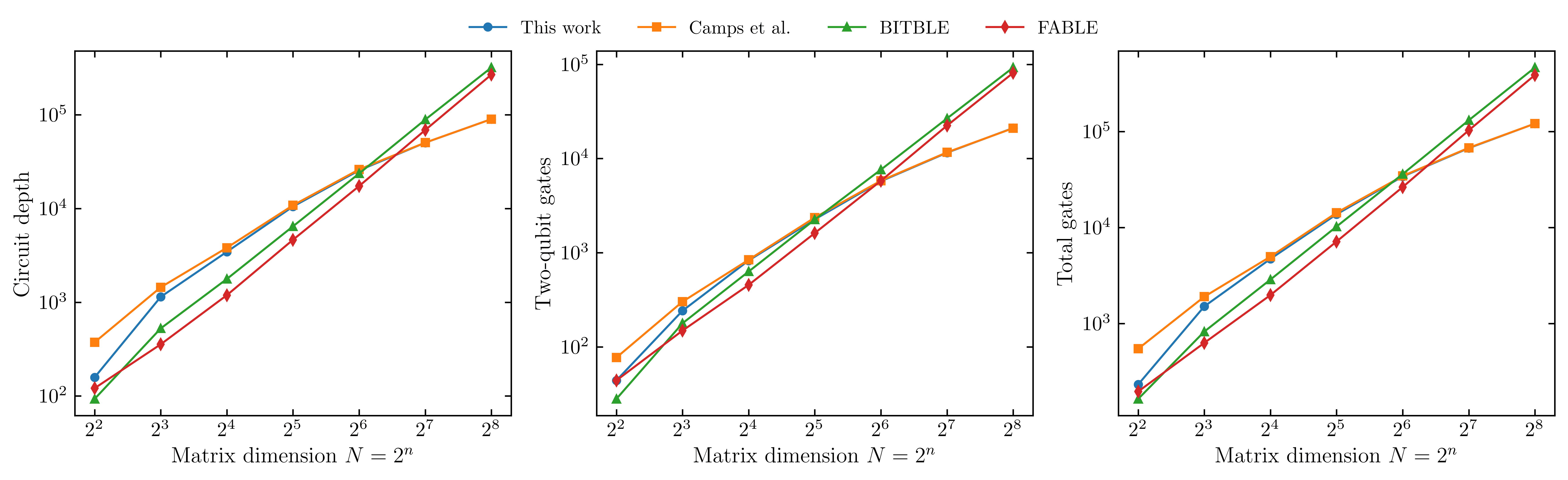}
    \caption{Comparison of resources for block encoding constructions of the $1$D discrete Laplacian with periodic boundary conditions}
    \label{fig:1D_C_P}
\end{figure}

\begin{figure}[htbp]
    \centering
    \includegraphics[width=1\linewidth]{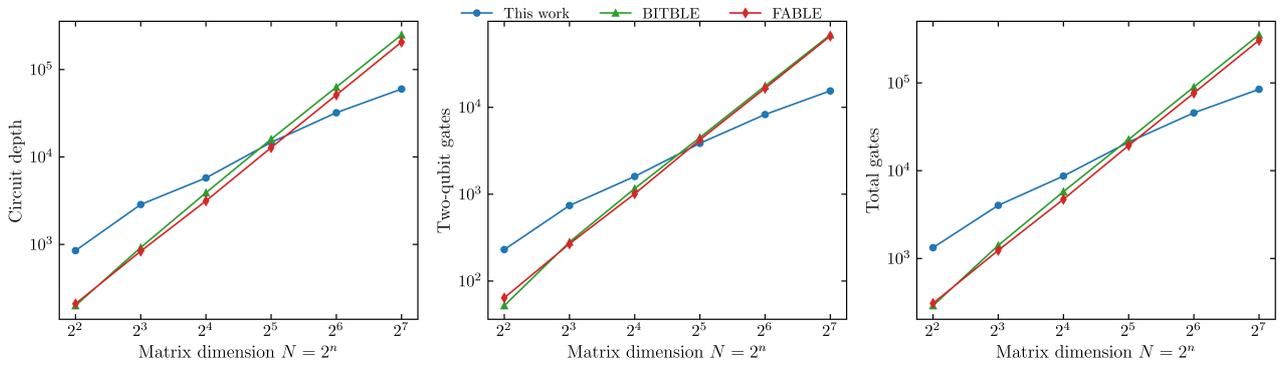}
    \caption{Comparison of resources for block encoding constructions of the $1$D discrete Laplacian with Neumann boundary conditions.
}
    \label{fig:1D_FB_N}
\end{figure}

\begin{figure}[htbp]
    \centering
    \includegraphics[width=1\linewidth]{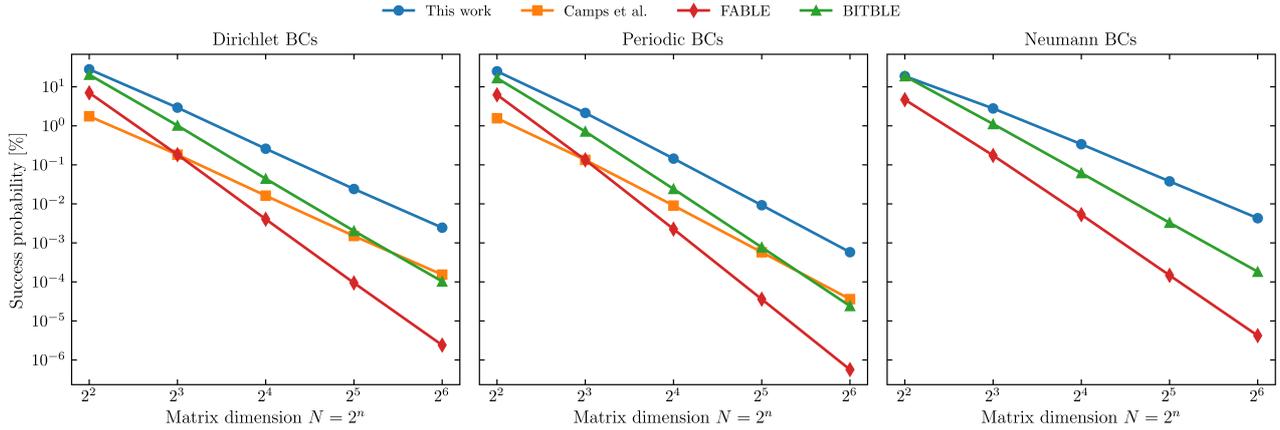}
    \caption{Success probability of different block encoding constructions for the $1$D discrete Laplacian under Dirichlet, periodic, and Neumann boundary conditions.
}
    \label{fig:SP_1D}
\end{figure}

\subsubsection{$2$D Laplacians}

For the Dirichlet--Dirichlet case shown in Fig.~\ref{fig:2D_S_DD}, the proposed construction exhibits substantially lower circuit depth, two-qubit gate counts, and total gate counts compared with the constructions in~\cite{S_nderhauf_2024},~\cite{camps2022fable}, and~\cite{BITBLE}. Fig.~\ref{fig:SP_2D_DD}  reports the corresponding success probabilities. The construction proposed in this work maintains higher success probabilities than the alternative approaches, with significantly slower decay rate.

Appendix~\ref{app:be_additional}, shows similar results with mixed periodic--Neumann boundary conditions.

\begin{figure}[htbp]
    \centering
    \includegraphics[width=\linewidth]{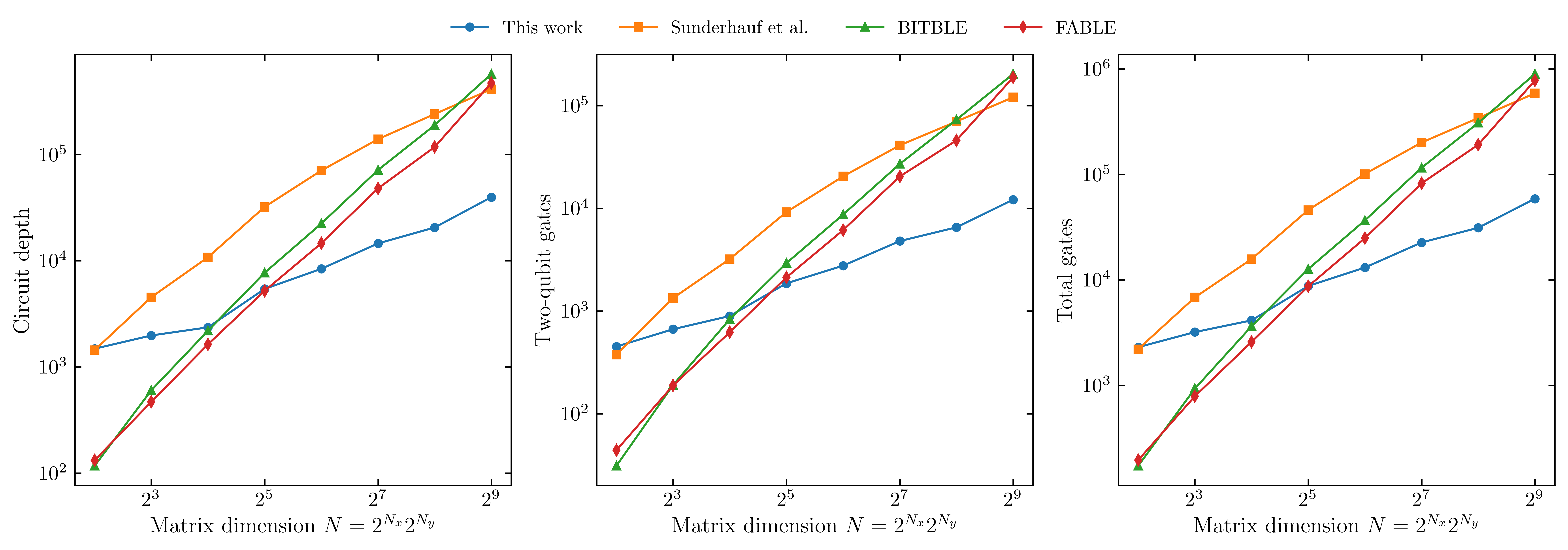}
    \caption{Comparison of resources for block encoding constructions of the $2$D Laplacian with Dirichlet boundary conditions applied along both spatial directions.}
    \label{fig:2D_S_DD}
\end{figure}

\begin{figure}[htbp]
    \centering
    \includegraphics[width=0.6\linewidth]{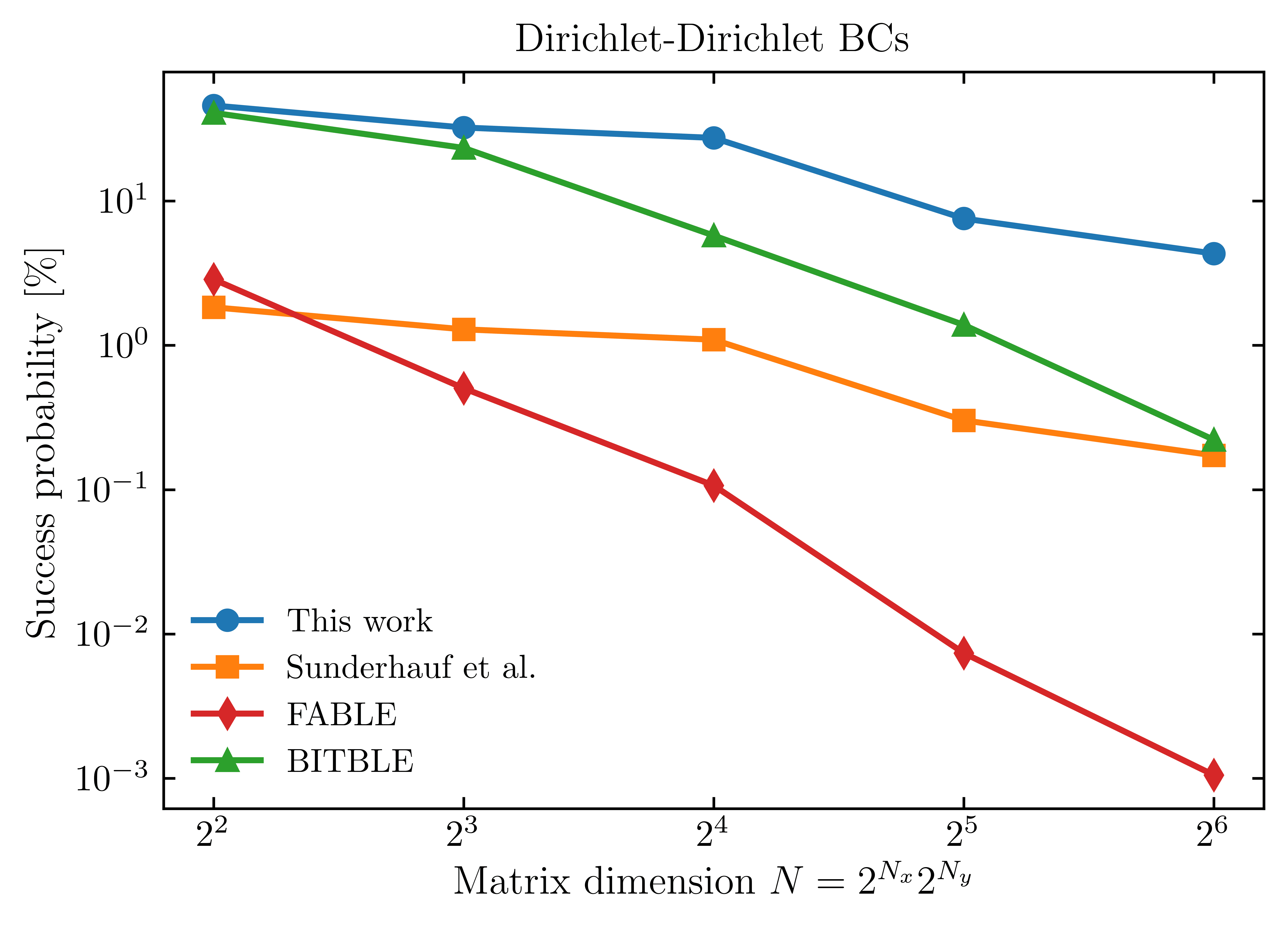}
    \caption{Success probability of block encoding constructions for the $2$D Laplacian with Dirichlet boundary conditions applied along both spatial directions.}
    \label{fig:SP_2D_DD}
\end{figure}

\subsubsection{$3$D Laplacians}

Finally, we illustrate the behavior of the proposed construction for a $3$D Laplacian operator. As a representative example, we consider a Laplacian with periodic, Dirichlet, and Neumann boundary conditions applied along the three spatial directions.

The results in Fig.~\ref{fig:3D_FB_PDN} demonstrate that the proposed construction yields significantly smaller circuits compared with the explicit block encoding circuits of~\cite{camps2022fable}, and~\cite{BITBLE}. The gap between the methods becomes more pronounced as the problem dimension grows, reflecting the advantage of exploiting the structured form of the Laplacian operator rather than explicitly encoding the full matrix.

Fig.~\ref{fig:SP_3D_PDN} shows the corresponding success probabilities for the same example. As in the lower-dimensional cases, the proposed method consistently maintains higher success probabilities and exhibits a significantly slower decay rates.

We emphasize that this three-dimensional case serves as an illustrative example and the proposed framework naturally extends to higher dimensions and supports heterogeneous boundary conditions without substantial increases in circuit complexity.

\begin{figure}[htbp]
    \centering
    \includegraphics[width=\linewidth]{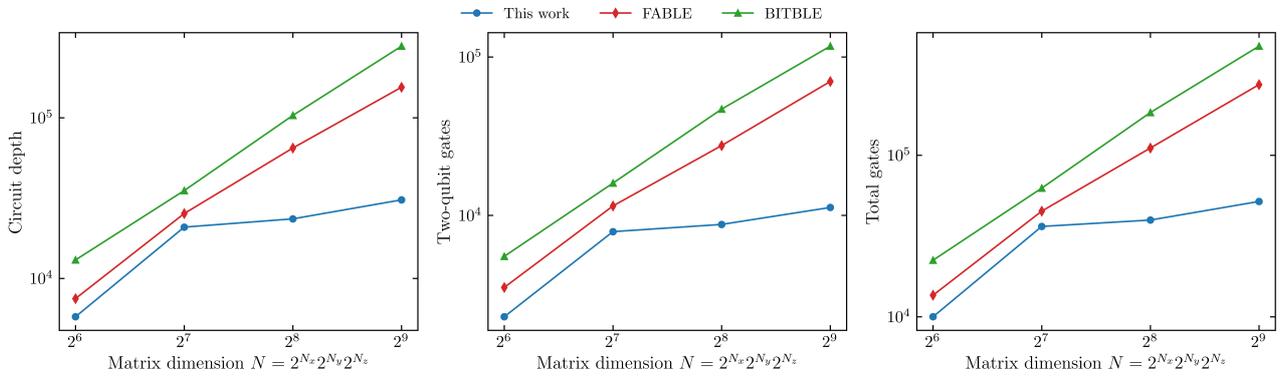}
    \caption{Comparison of resources for block encoding constructions of a $3$D Laplacian with periodic, Dirichlet, and Neumann boundary conditions applied along the three spatial directions.}
    \label{fig:3D_FB_PDN}
\end{figure}

\begin{figure}[htbp]
    \centering
    \includegraphics[width=0.6\linewidth]{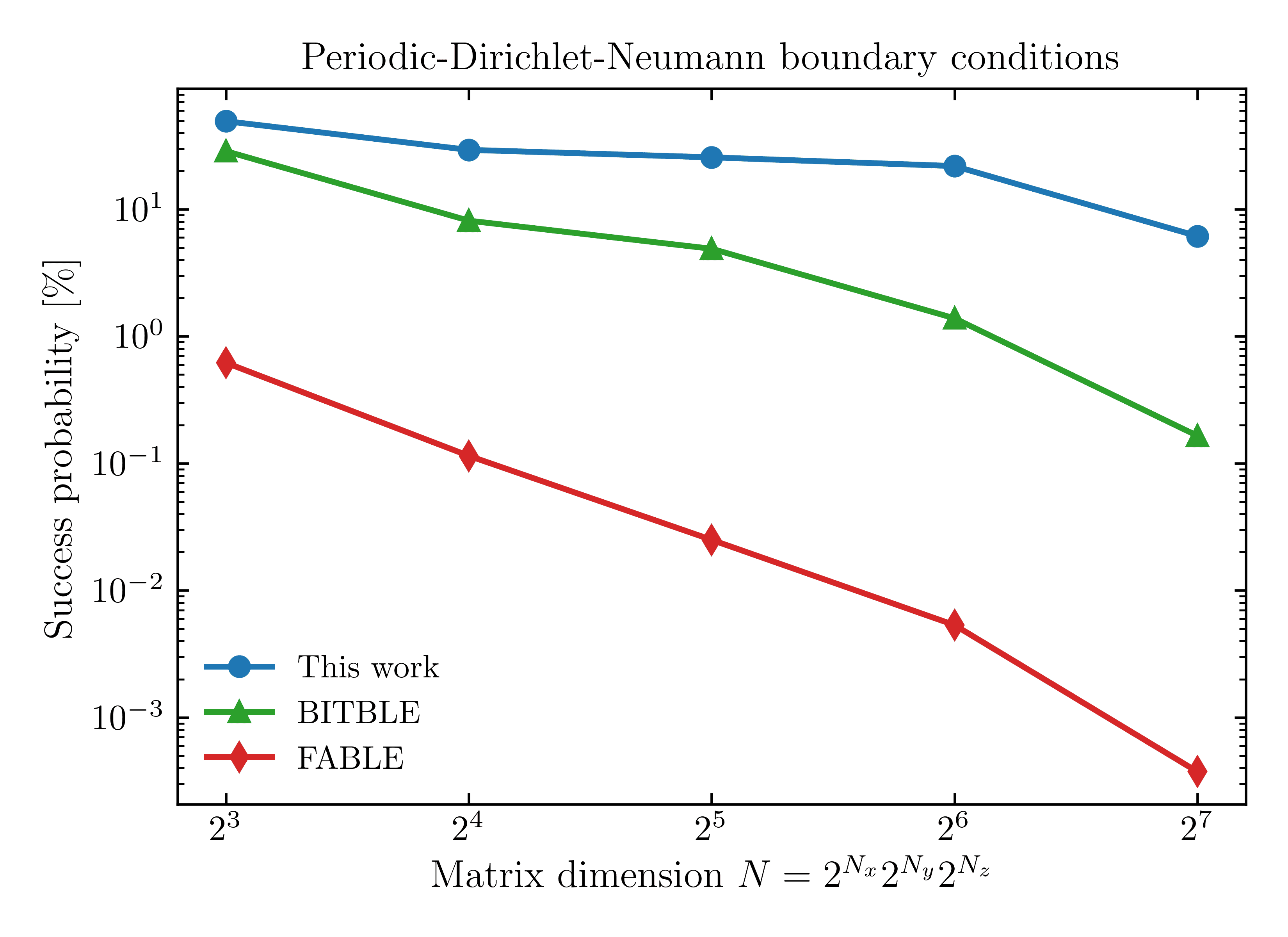}
    \caption{Success probability of block encoding constructions for the $3$D Laplacian with periodic, Dirichlet, and Neumann boundary conditions.}
    \label{fig:SP_3D_PDN}
\end{figure}

\section{Conclusion}
\label{conclusion}

In this work, we presented a unified framework for explicitly block encoding discrete Laplacian operators arising from finite-difference discretizations. Previous structure and sparsity aware explicit constructions addressed only specific settings such as $1$D Laplacians with Dirichlet or periodic boundary conditions, $2$D Laplacians with Dirichlet boundaries, or $D$-dimensional Laplacians with fully periodic boundaries and uniform grids. Moreover, these constructions were developed independently for particular configurations and were not formulated within a common framework that naturally extends across boundary conditions and spatial dimensions. 

The framework introduced in this work bridges this gap by providing a single modular construction that supports Dirichlet, periodic, and Neumann boundary conditions in arbitrary spatial dimensions while allowing different boundary conditions and grid sizes to be specified independently along each coordinate axis. Analytical resource estimates show that the proposed construction achieves favorable scaling. In particular, circuits obtained from our construction exhibit consistently lower depth and gate counts when transpiled to an IBM quantum hardware gate set, while simulations using the \texttt{Aer} statevector backend indicate higher success probabilities across the tested problem sizes.

The framework introduced here opens several directions for future work. One natural extension is the incorporation of more general boundary conditions, such as mixed Robin boundary conditions, within the same block encoding architecture. Another promising direction is the application of similar structure exploiting techniques to other classes of structured sparse operators beyond Laplacians, including operators arising from higher-order differential discretizations or graph-based models. Finally, it will be important to study how the present block encodings impact the performance of downstream quantum algorithms, particularly quantum singular value transformation (QSVT) based methods for solving linear system of equations.

\begin{acknowledgments}
This work was carried out as part of the Qiskit Advocate Mentorship Program 2025, organized within the IBM Qiskit Advocate Program. The authors thank the organizers for facilitating the program and supporting collaboration within the Qiskit community.
\end{acknowledgments}

\bibliographystyle{quantum}
\bibliography{references} 

\begin{thebibliography}{10}

\bibitem{smith1993numerical}
G.~D. Smith.
\newblock ``Numerical solution of partial differential equations: Finite difference methods''.
\newblock Oxford Applied Mathematics and Computing Science Series. Clarendon Press. Oxford~(1993).
\newblock 3 edition.

\bibitem{krylov}
Tomohiro Sogabe.
\newblock ``Krylov subspace methods for linear systems''.
\newblock Volume~60 of Springer Series in Computational Mathematics.
\newblock Springer Nature. Singapore~(2022).

\bibitem{NielsenChuang2002}
Michael~A. Nielsen and Isaac~L. Chuang.
\newblock ``Quantum computation and quantum information''.
\newblock \href{https://dx.doi.org/10.1017/CBO9780511976667}{Cambridge University Press}. Cambridge~(2010).

\bibitem{morales2025quantumlinearsolverssurvey}
Mauro E.~S. Morales, Lirandë Pira, Philipp Schleich, Kelvin Koor, Pedro C.~S. Costa, Dong An, Alán Aspuru-Guzik, Lin Lin, Patrick Rebentrost, and Dominic~W. Berry.
\newblock ``Quantum linear system solvers: A survey of algorithms and applications''~(2025).
\newblock  \href{http://arxiv.org/abs/2411.02522}{arXiv:2411.02522}.

\bibitem{qspLow}
Guang~Hao Low and Isaac~L. Chuang.
\newblock ``Optimal hamiltonian simulation by quantum signal processing''.
\newblock \href{https://dx.doi.org/10.1103/PhysRevLett.118.010501}{Phys. Rev. Lett. {\bf 118}, 010501}~(2017).

\bibitem{Gily_n_2019}
András Gilyén, Yuan Su, Guang~Hao Low, and Nathan Wiebe.
\newblock ``Quantum singular value transformation and beyond: exponential improvements for quantum matrix arithmetics''.
\newblock In Proceedings of the 51st Annual ACM SIGACT Symposium on Theory of Computing.
\newblock \href{https://dx.doi.org/10.1145/3313276.3316366}{Page 193–204}.
\newblock STOC ’19. ACM~(2019).

\bibitem{camps2022fable}
Daan Camps and Roel Van~Beeumen.
\newblock ``{FABLE}: Fast approximate quantum circuits for block-encodings''~(2022).
\newblock  \href{http://arxiv.org/abs/2205.00081}{arXiv:2205.00081}.

\bibitem{qram}
B.~David Clader, Alexander~M. Dalzell, Nikitas Stamatopoulos, Grant Salton, Mario Berta, and William~J. Zeng.
\newblock ``Quantum resources required to block-encode a matrix of classical data''.
\newblock \href{https://dx.doi.org/10.1109/TQE.2022.3231194}{IEEE Transactions on Quantum Engineering {\bf 3}, 1--23}~(2022).

\bibitem{lcu_childs}
Andrew~M. Childs and Nathan Wiebe.
\newblock ``Hamiltonian simulation using linear combinations of unitary operations''.
\newblock Quantum Info. Comput. {\bf 12}, 901–924~(2012).

\bibitem{BITBLE}
Zexian Li, Xiao-Ming Zhang, Chunlin Yang, and Guofeng Zhang.
\newblock ``Binary tree block encoding of classical matrix''.
\newblock \href{https://dx.doi.org/10.1109/TQE.2025.3624699}{IEEE Transactions on Quantum Engineering {\bf 7}, 1--18}~(2026).

\bibitem{kuklinski2025efficientblockencodingsrequirestructure}
Parker Kuklinski, Benjamin Rempfer, Justin Elenewski, and Kevin Obenland.
\newblock ``Efficient block-encodings require structure''~(2025).
\newblock  \href{http://arxiv.org/abs/2509.19667}{arXiv:2509.19667}.

\bibitem{low2019hamiltonian}
Guang~Hao Low and Isaac~L. Chuang.
\newblock ``Hamiltonian simulation by qubitization''.
\newblock Quantum {\bf 3}, 163~(2019).

\bibitem{camps_explicit_2023}
Daan Camps, Lin Lin, Roel Van~Beeumen, and Chao Yang.
\newblock ``Explicit quantum circuits for block encodings of certain sparse matrices''~(2023).
\newblock  \href{http://arxiv.org/abs/2203.10236}{arXiv:2203.10236}.

\bibitem{S_nderhauf_2024}
Christoph Sünderhauf, Earl Campbell, and Joan Camps.
\newblock ``Block-encoding structured matrices for data input in quantum computing''.
\newblock \href{https://dx.doi.org/10.22331/q-2024-01-11-1226}{Quantum {\bf 8}, 1226}~(2024).

\bibitem{sturm2025efficientexplicitblockencoding}
Andreas Sturm and Niclas Schillo.
\newblock ``Efficient and explicit block encoding of finite difference discretizations of the laplacian''~(2025).
\newblock  \href{http://arxiv.org/abs/2509.02429}{arXiv:2509.02429}.

\bibitem{javadiabhari2024quantumcomputingqiskit}
Ali Javadi-Abhari, Matthew Treinish, Kevin Krsulich, Christopher~J. Wood, Jake Lishman, Julien Gacon, Simon Martiel, Paul~D. Nation, Lev~S. Bishop, Andrew~W. Cross, Blake~R. Johnson, and Jay~M. Gambetta.
\newblock ``Quantum computing with qiskit''~(2024).
\newblock  \href{http://arxiv.org/abs/2405.08810}{arXiv:2405.08810}.

\bibitem{tdc28_qamp2025}
TDC28.
\newblock ``laplacian\_beqc''.
\newblock \url{https://github.com/TDC28/laplacian_beqc}~(2025).
\newblock GitHub repository.

\bibitem{grover2002creatingsuperpositionscorrespondefficiently}
Lov Grover and Terry Rudolph.
\newblock ``Creating superpositions that correspond to efficiently integrable probability distributions''~(2002).
\newblock  \href{http://arxiv.org/abs/quant-ph/0208112}{arXiv:quant-ph/0208112}.

\bibitem{KhattarGidney2025}
Tanuj Khattar and Craig Gidney.
\newblock ``Rise of conditionally clean ancillae for efficient quantum circuit constructions''.
\newblock \href{https://dx.doi.org/10.22331/q-2025-05-21-1752}{Quantum {\bf 9}, 1752}~(2025).

\bibitem{barenco_1995}
Adriano Barenco, Charles~H. Bennett, Richard Cleve, David~P. DiVincenzo, Norman Margolus, Peter Shor, Tycho Sleator, John~A. Smolin, and Harald Weinfurter.
\newblock ``Elementary gates for quantum computation''.
\newblock \href{https://dx.doi.org/10.1103/PhysRevA.52.3457}{Phys. Rev. A {\bf 52}, 3457--3467}~(1995).

\bibitem{mcu_2025}
Ben Zindorf and Sougato Bose.
\newblock ``Efficient implementation of multicontrolled quantum gates''.
\newblock \href{https://dx.doi.org/10.1103/8blx-nfcr}{Phys. Rev. Appl. {\bf 24}, 044030}~(2025).

\bibitem{selinger_2015}
Peter Selinger.
\newblock ``Quantum circuits of {T}-depth one''.
\newblock \href{https://dx.doi.org/10.1103/PhysRevA.87.042302}{Phys. Rev. A {\bf 87}, 042302}~(2013).

\bibitem{Gidney2015ControlledNots}
Craig Gidney.
\newblock ``Constructing large controlled {NOT}s''.
\newblock \url{https://algassert.com/circuits/2015/06/05/Constructing-Large-Controlled-Nots.html}~(2015).
\newblock Blog post.

\bibitem{rel_phase_toffoli}
Dmitri Maslov.
\newblock ``Advantages of using relative-phase toffoli gates with an application to multiple control toffoli optimization''.
\newblock \href{https://dx.doi.org/10.1103/PhysRevA.93.022311}{Phys. Rev. A {\bf 93}, 022311}~(2016).

\bibitem{adders_1996}
Vlatko Vedral, Adriano Barenco, and Artur Ekert.
\newblock ``Quantum networks for elementary arithmetic operations''.
\newblock \href{https://dx.doi.org/10.1103/PhysRevA.54.147}{Phys. Rev. A {\bf 54}, 147--153}~(1996).

\bibitem{gidney2018halving}
Craig Gidney.
\newblock ``Halving the cost of quantum addition''.
\newblock \href{https://dx.doi.org/10.22331/q-2018-06-18-74}{Quantum {\bf 2}, 74}~(2018).

\bibitem{draper2004}
Thomas~G. Draper, Samuel~A. Kutin, Eric~M. Rains, and Krysta~M. Svore.
\newblock ``A logarithmic-depth quantum carry-lookahead adder''~(2004).
\newblock  \href{http://arxiv.org/abs/quant-ph/0406142}{arXiv:quant-ph/0406142}.

\bibitem{Low_2024}
Guang~Hao Low, Vadym Kliuchnikov, and Luke Schaeffer.
\newblock ``Trading {T} gates for dirty qubits in state preparation and unitary synthesis''.
\newblock \href{https://dx.doi.org/10.22331/q-2024-06-17-1375}{Quantum {\bf 8}, 1375}~(2024).

\end{thebibliography}
\newpage

\section*{Appendix}
\appendix

\section{Structure of Discrete Laplacian Matrices}
\label{app:matrix_structure}

In this appendix we visualize the structure of the discrete Laplacian operators considered throughout the paper. Rather than constructing the matrices directly, the color maps are obtained by extracting and post-processing the top-left block of the block encoding unitary produced by our circuit construction. The resulting color maps verify that the encoded operator exhibits the expected structure and sparsity patterns associated with finite-difference discretizations under different boundary conditions and spatial dimensions.

Fig.~\ref{fig:matrix_structures_1D},~\ref{fig:matrix_structures_2D} and~\ref{fig:matrix_structures_3D} show representative examples for one-, two-, and three-dimensional discretizations respectively used in the numerical experiments.

\begin{figure}[htbp]
\centering
\begin{subfigure}{0.32\linewidth}
\centering
\includegraphics[width=\linewidth]{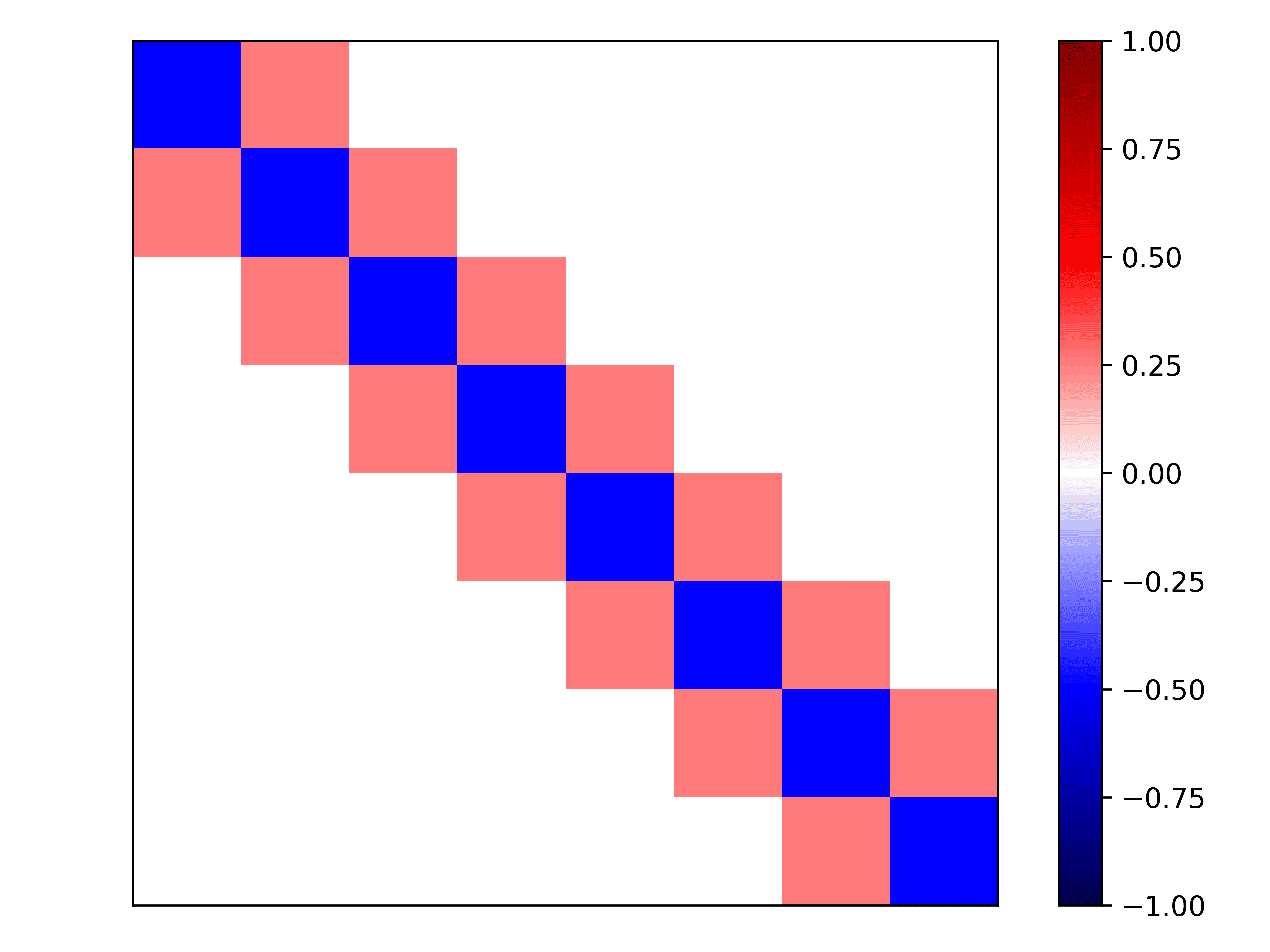}
\caption{1D Dirichlet}
\end{subfigure}
\hfill
\begin{subfigure}{0.32\linewidth}
\centering
\includegraphics[width=\linewidth]{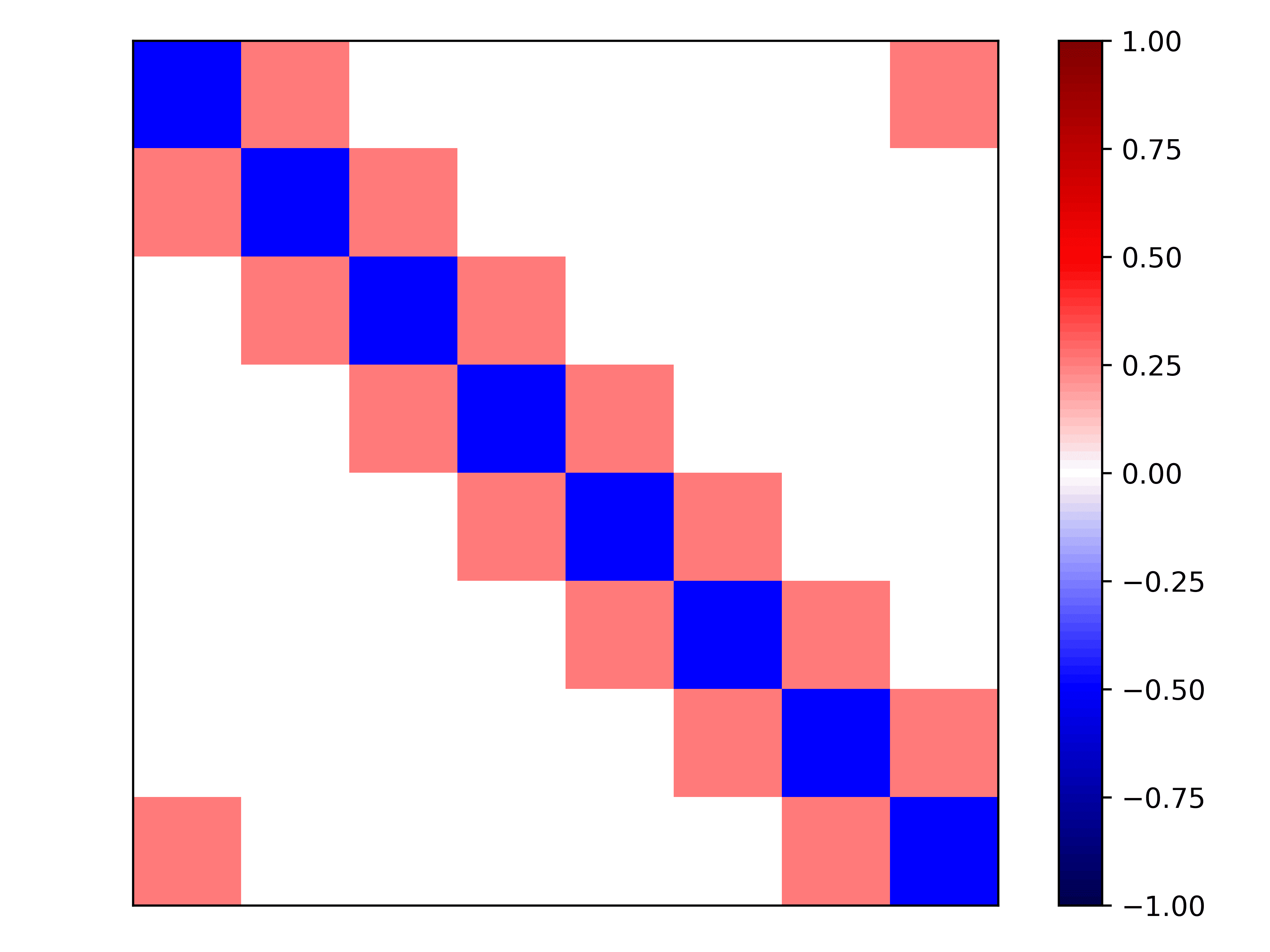}
\caption{1D periodic}
\end{subfigure}
\hfill
\begin{subfigure}{0.32\linewidth}
\centering
\includegraphics[width=\linewidth]{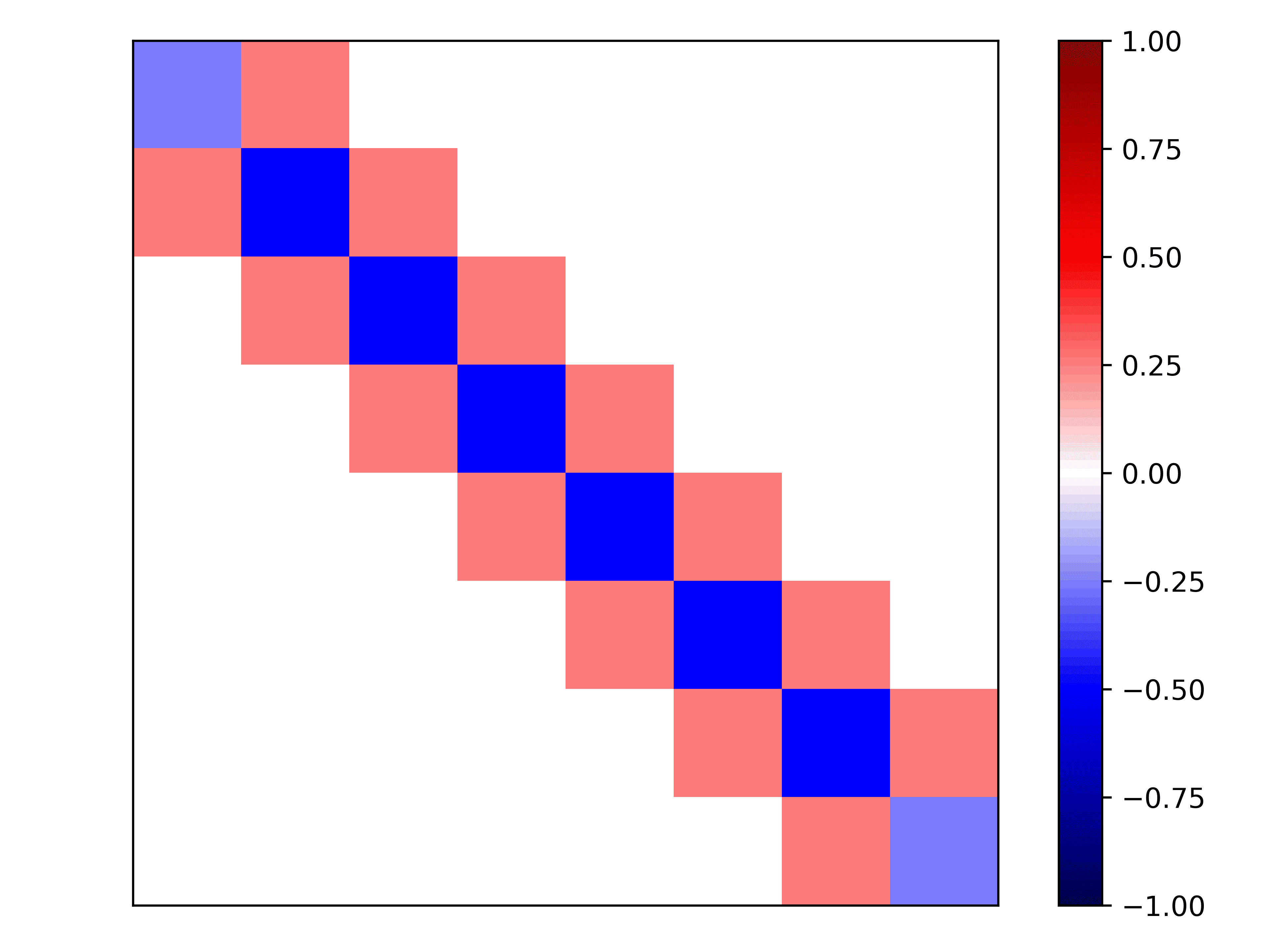}
\caption{1D Neumann}
\end{subfigure}
\caption{Heat maps of the matrices obtained from the top-left block of the block encoding unitaries $U_{\mathrm{d}}, U_{\mathrm{p}}, U_{\mathrm{n}}$.}
\label{fig:matrix_structures_1D}
\end{figure}

\begin{figure}[htbp]
\centering
\begin{subfigure}{0.44\linewidth}
\centering
\includegraphics[width=\linewidth]{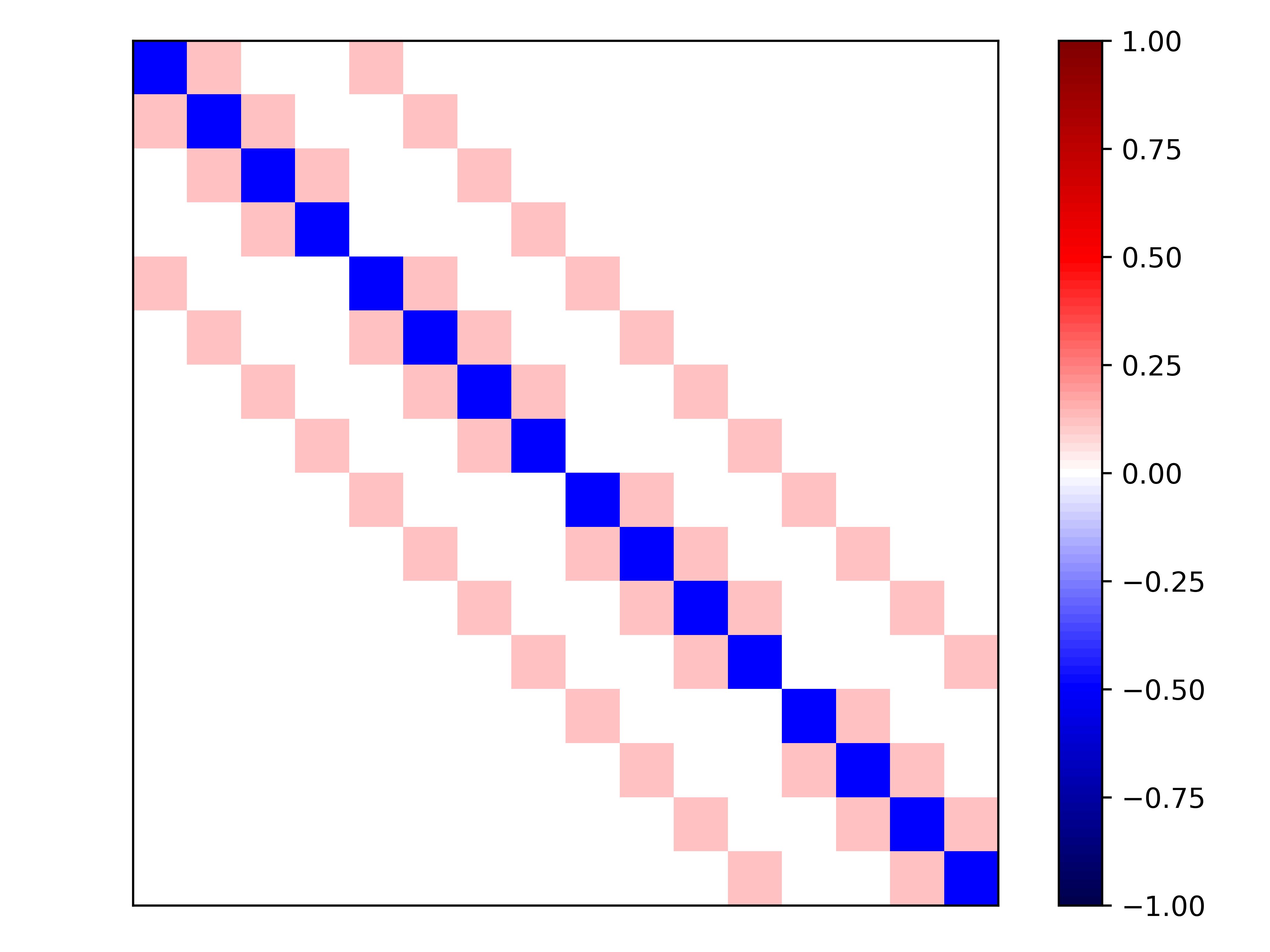}
\caption{2D Dirichlet--Dirichlet}
\end{subfigure}
\hfill
\begin{subfigure}{0.44\linewidth}
\centering
\includegraphics[width=\linewidth]{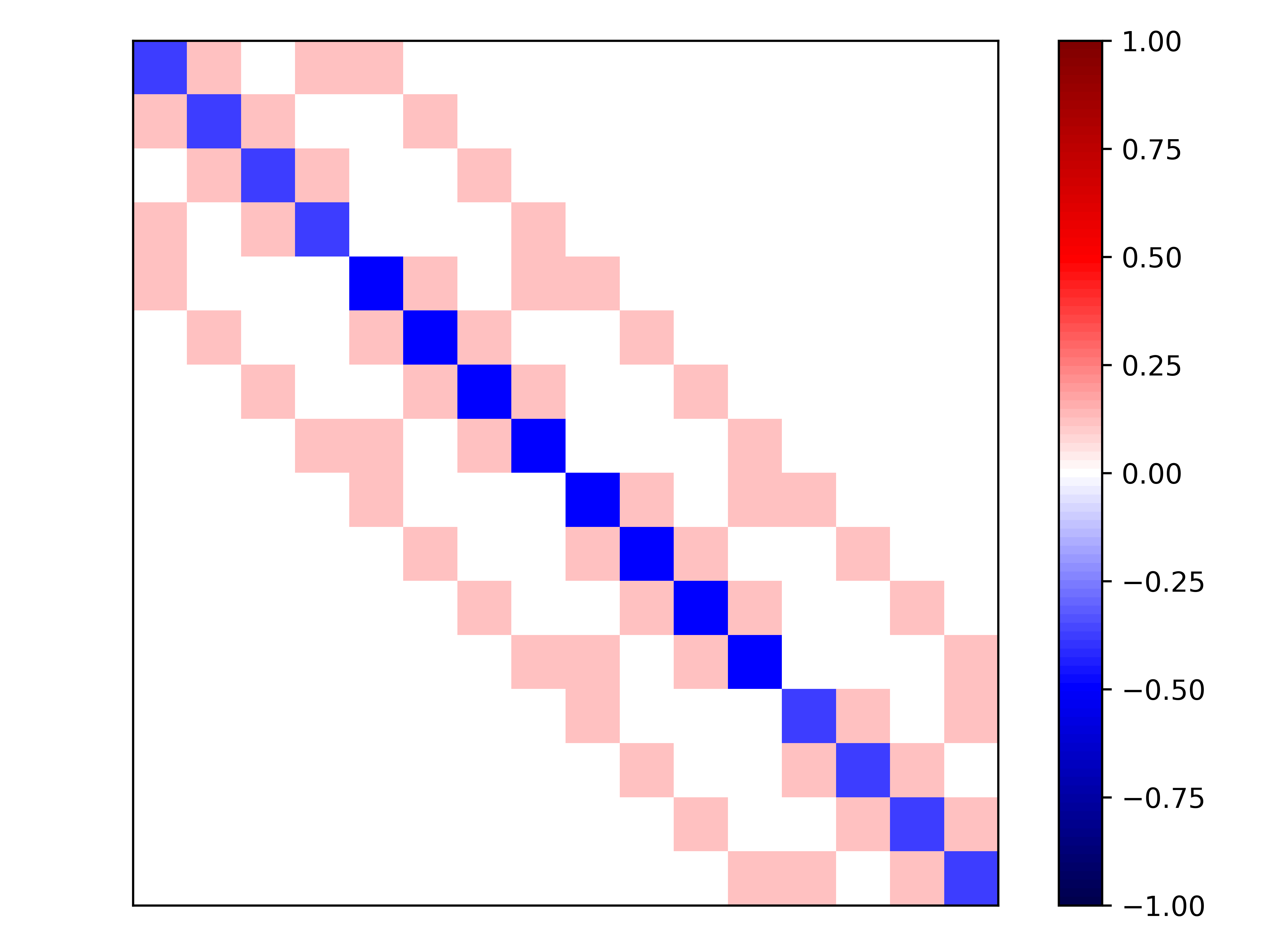}
\caption{2D periodic--Neumann}
\end{subfigure}
\caption{Heat maps of the matrices obtained from the top-left block of the block encoding unitary $U^2$ constructed in this work, for various boundary conditions considered.}
\label{fig:matrix_structures_2D}
\end{figure}

\begin{figure}[htbp]
\centering
\begin{subfigure}{0.32\linewidth}
\centering
\includegraphics[width=\linewidth]{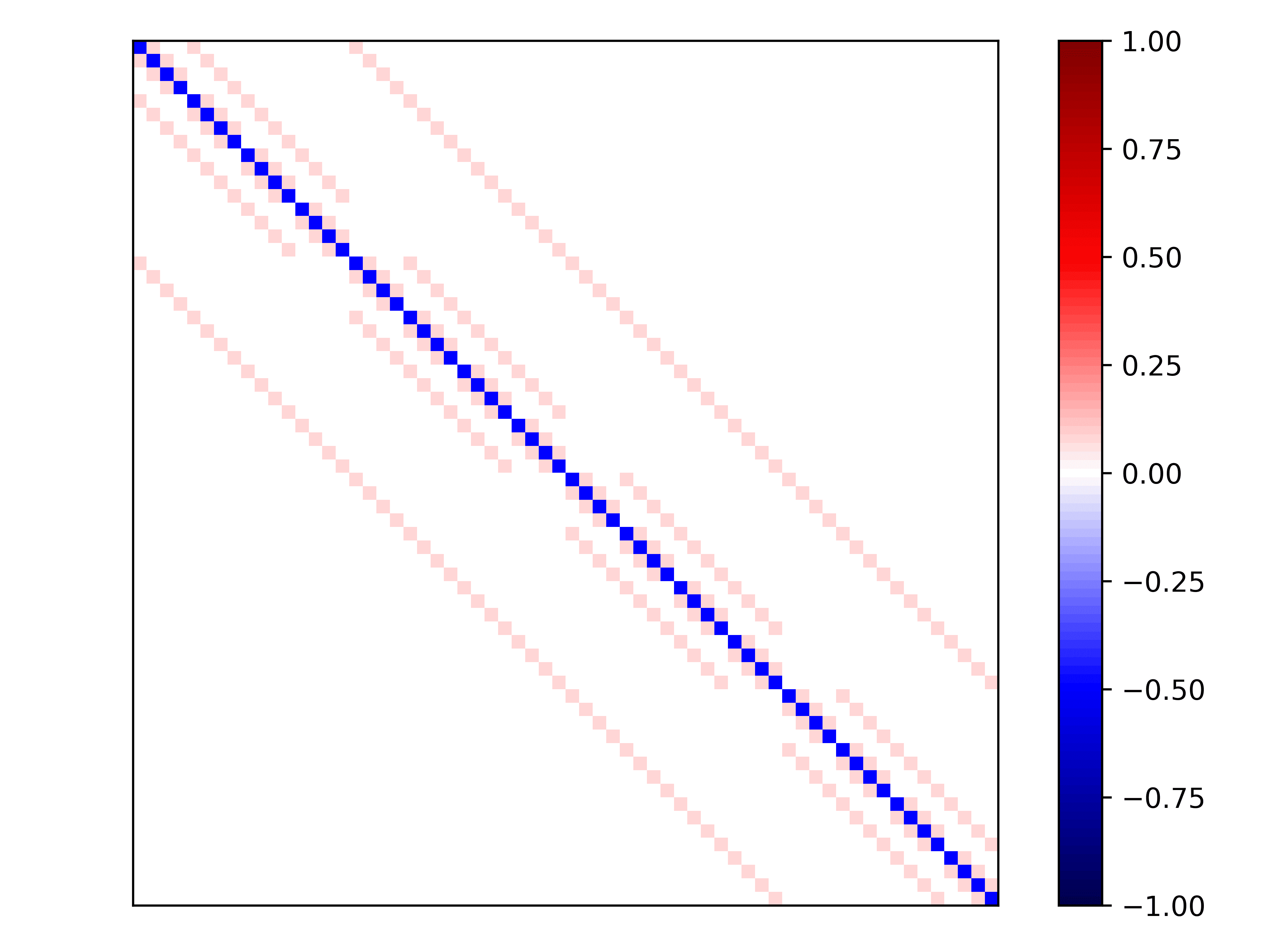}
\caption{3D Dirichlet--Dirichlet--Dirichlet}
\end{subfigure}
\hfill
\begin{subfigure}{0.32\linewidth}
\centering
\includegraphics[width=\linewidth]{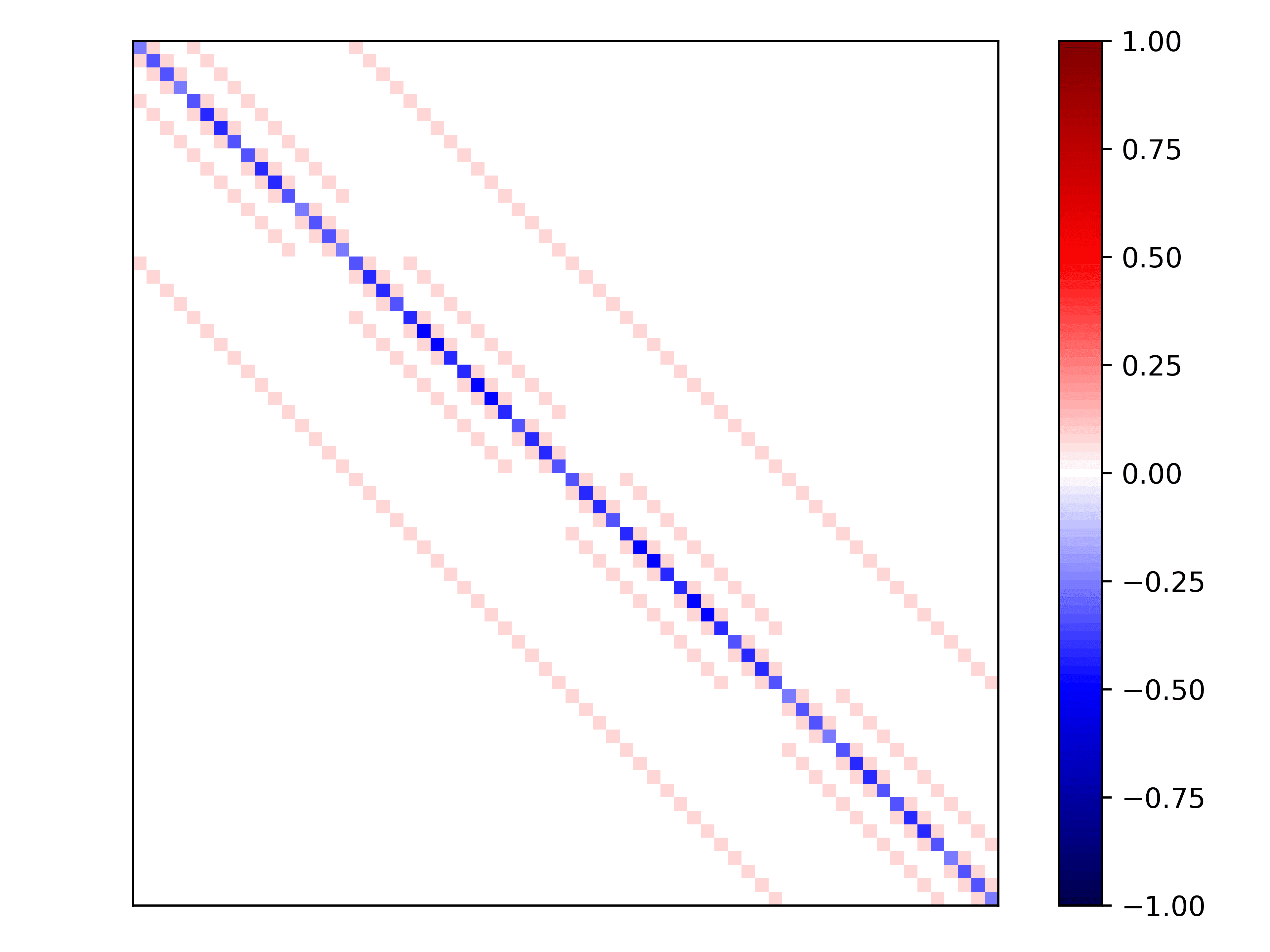}
\caption{3D Neumann--Neumann--Neumann}
\end{subfigure}
\hfill
\begin{subfigure}{0.32\linewidth}
\centering
\includegraphics[width=\linewidth]{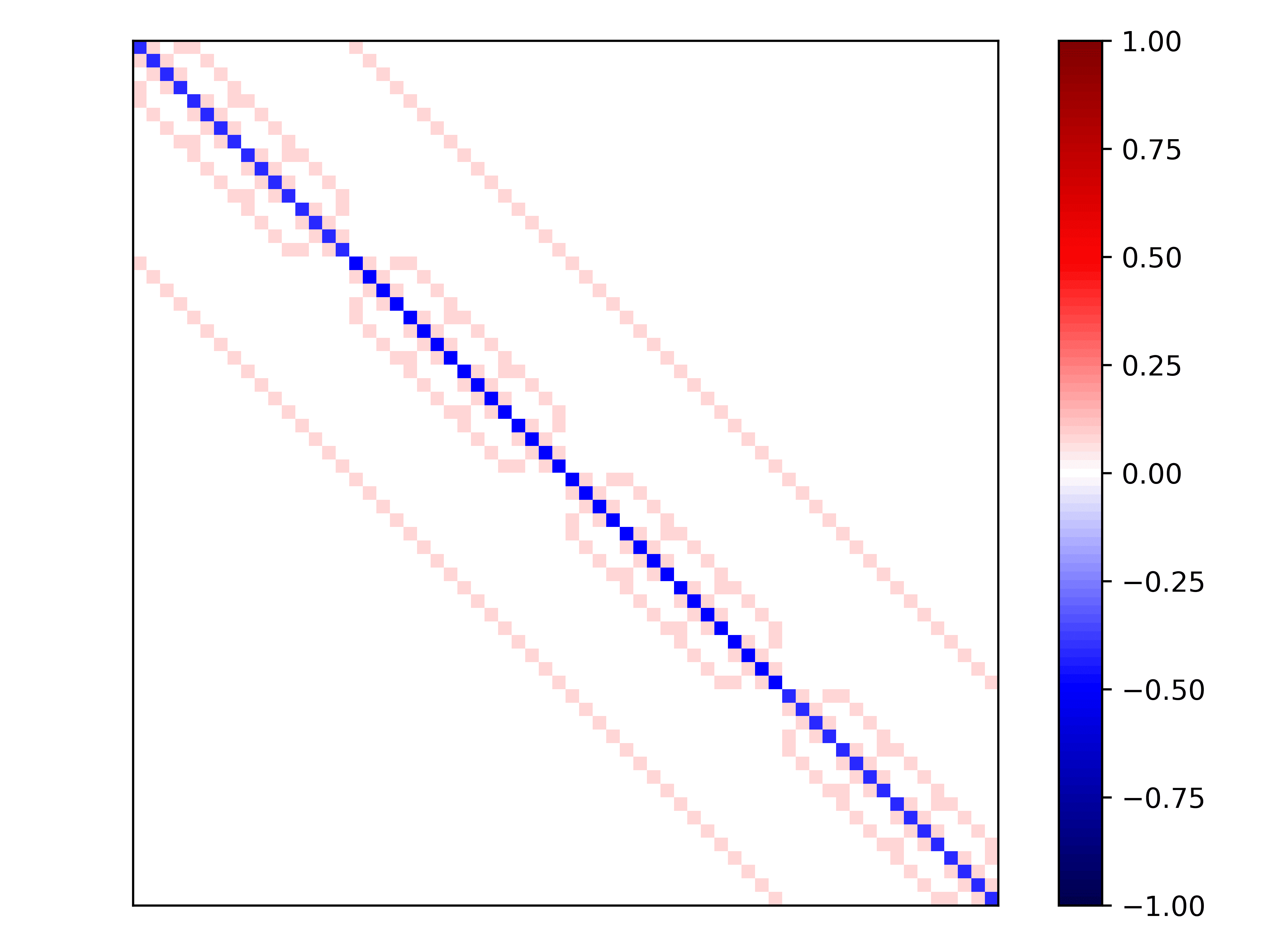}
\caption{3D periodic--Dirichlet--Neumann}
\end{subfigure}
\caption{Heat maps of the matrices obtained from the top-left block of the block encoding unitary $U^3$ constructed in this work, for various boundary conditions considered.}
\label{fig:matrix_structures_3D}
\end{figure}

\section{Additional Comparisons}
\label{app:be_additional}

This appendix presents additional benchmarking results that complement the analysis in the main text.

Figure~\ref{fig:2D_FB_PN} shows circuit-resource scaling for the $2$D Laplacian with mixed periodic--Neumann boundary conditions. The corresponding success probabilities are reported in Fig.~\ref{fig:SP_2D_PN}. To further illustrate the behavior of the proposed framework in higher dimensions, Fig.~\ref{fig:3D_DDD} presents benchmarking results for the $3$D Laplacian with Dirichlet boundary conditions along all spatial directions.

\begin{figure}[htbp]
    \centering
    \includegraphics[width=\linewidth]{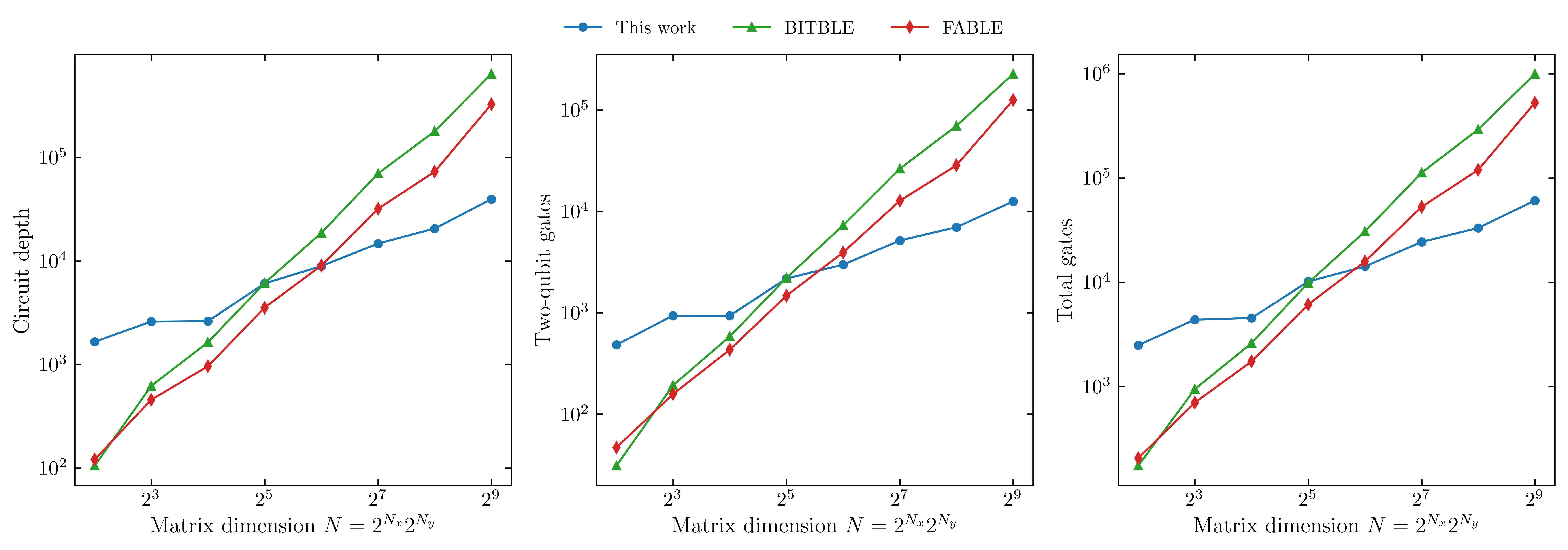}
    \caption{Comparison of circuit resources for block encoding constructions of the $2$D Laplacian with periodic boundary conditions in one spatial direction and Neumann boundary conditions in the other.}
    \label{fig:2D_FB_PN}
\end{figure}

\begin{figure}[htbp]
    \centering
    \includegraphics[width=0.6\linewidth]{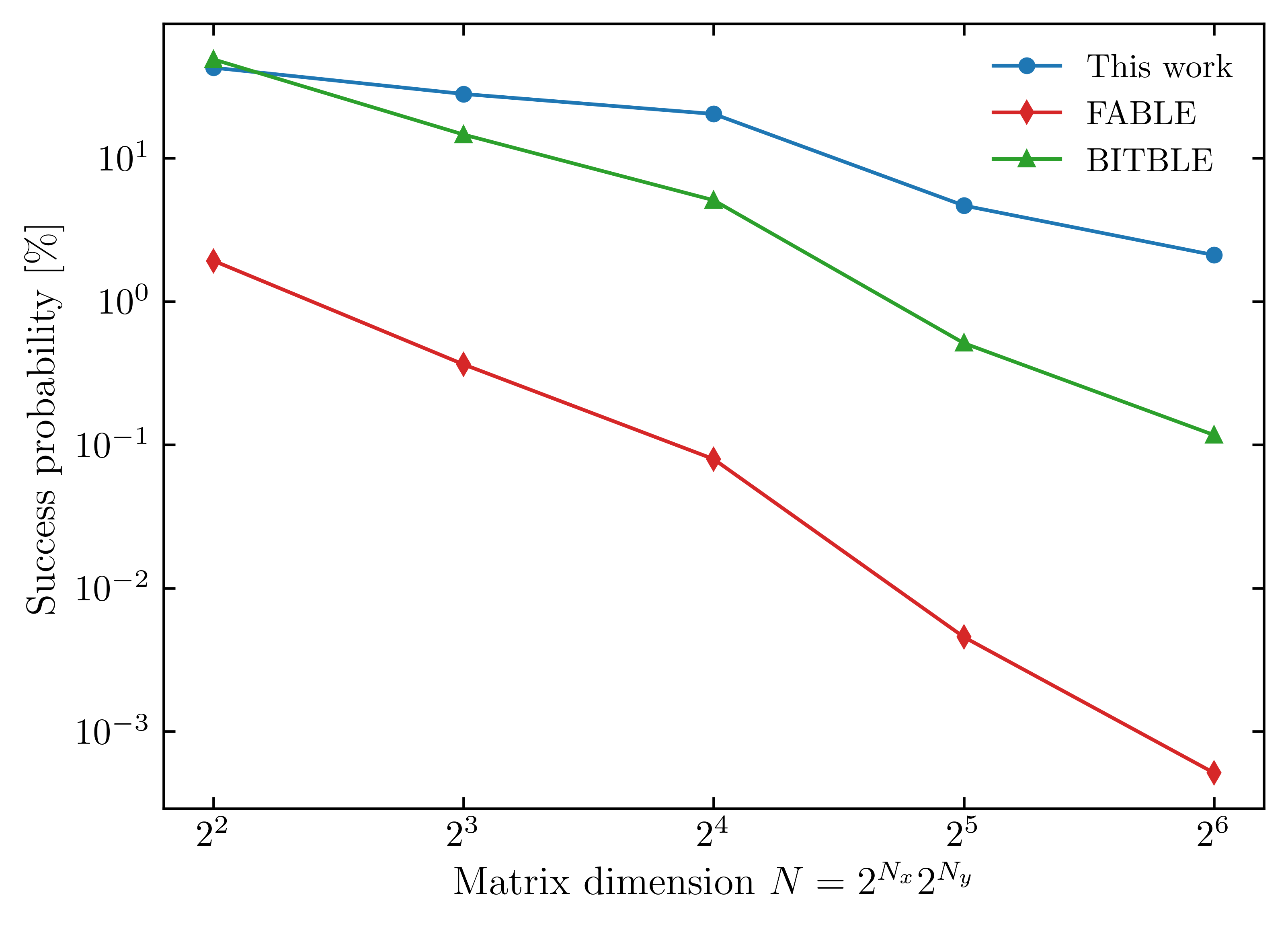}
    \caption{Success probability of the block encoding constructions for the $2$D Laplacian with periodic--Neumann boundary conditions.}
    \label{fig:SP_2D_PN}
\end{figure}

\begin{figure}[htbp]
    \centering
    \includegraphics[width=\linewidth]{figs/3d_laplacian_ddd_comparison.png}
    \caption{Comparison of circuit resources for block encoding constructions of the $3$D Laplacian with Dirichlet boundary conditions along all spatial directions.}
    \label{fig:3D_DDD}
\end{figure}

\begin{figure}[htbp]
    \centering
    \includegraphics[width=\linewidth]{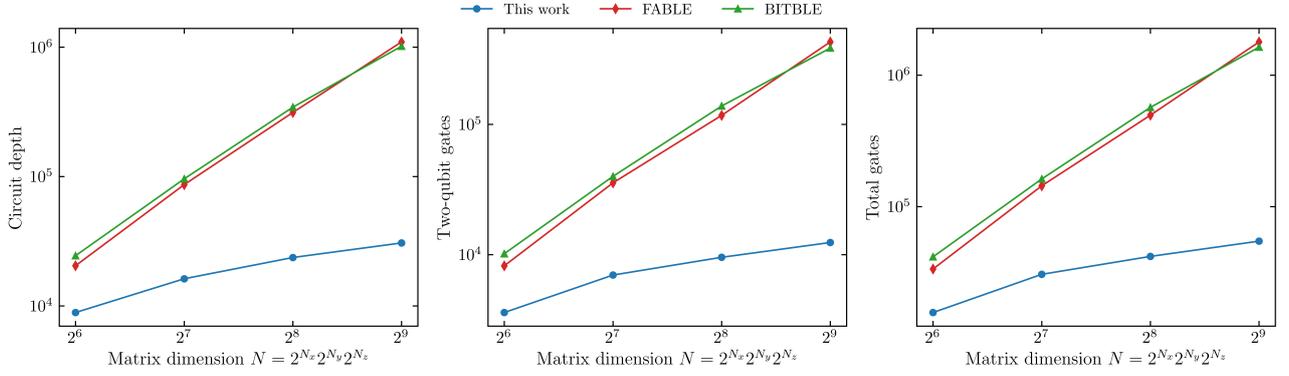}
    \caption{Comparison of circuit resources for block encoding constructions of the $3$D Laplacian with Neumann boundary conditions along all spatial directions.}
    \label{fig:3D_NNN}
\end{figure}

\begin{figure}[htbp]
\centering
\begin{subfigure}{0.48\linewidth}
\centering
\includegraphics[width=\linewidth]{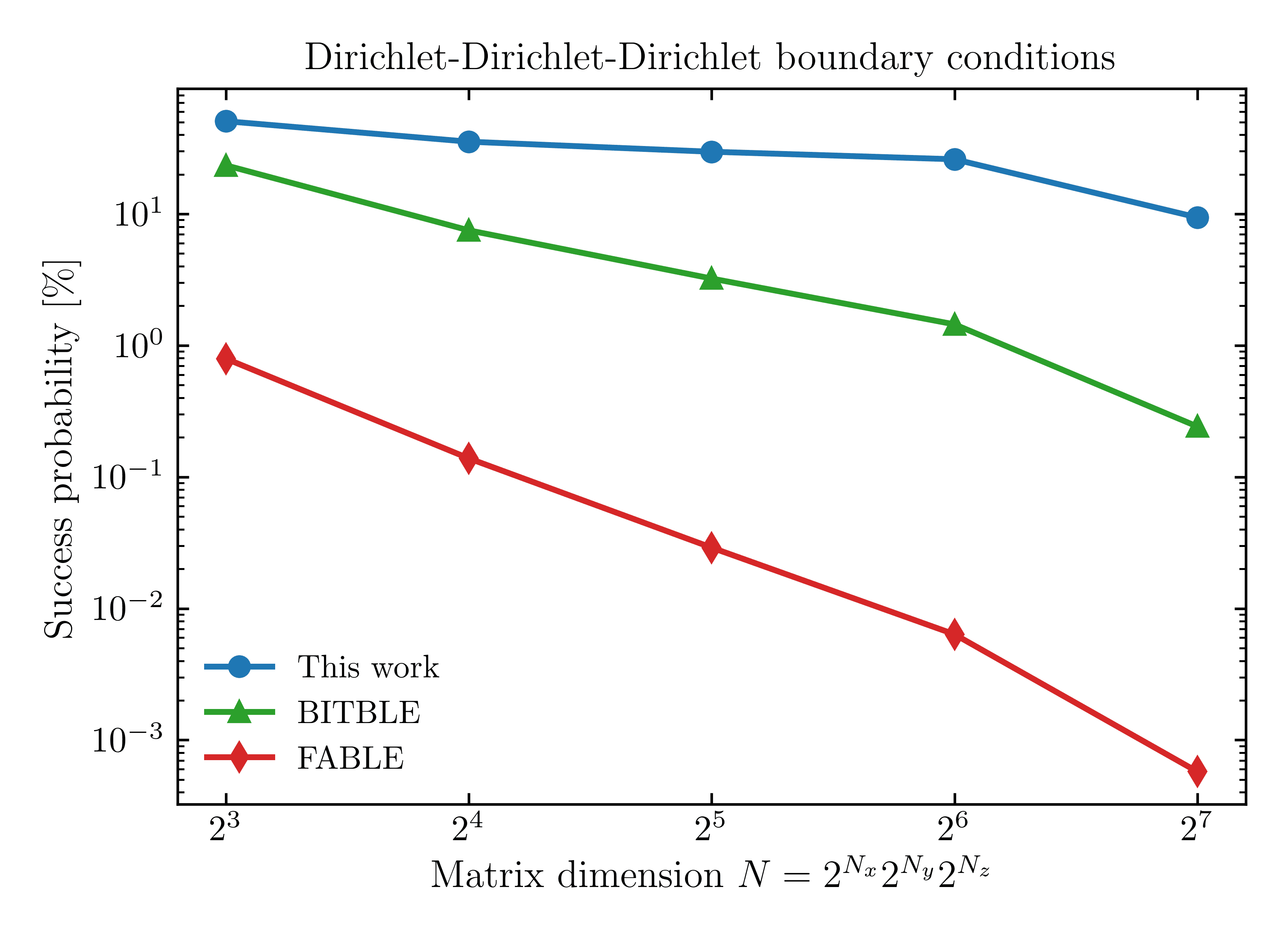}
\caption{Dirichlet-Dirichlet-Dirichlet}
\end{subfigure}
\hfill
\begin{subfigure}{0.48\linewidth}
\centering
\includegraphics[width=\linewidth]{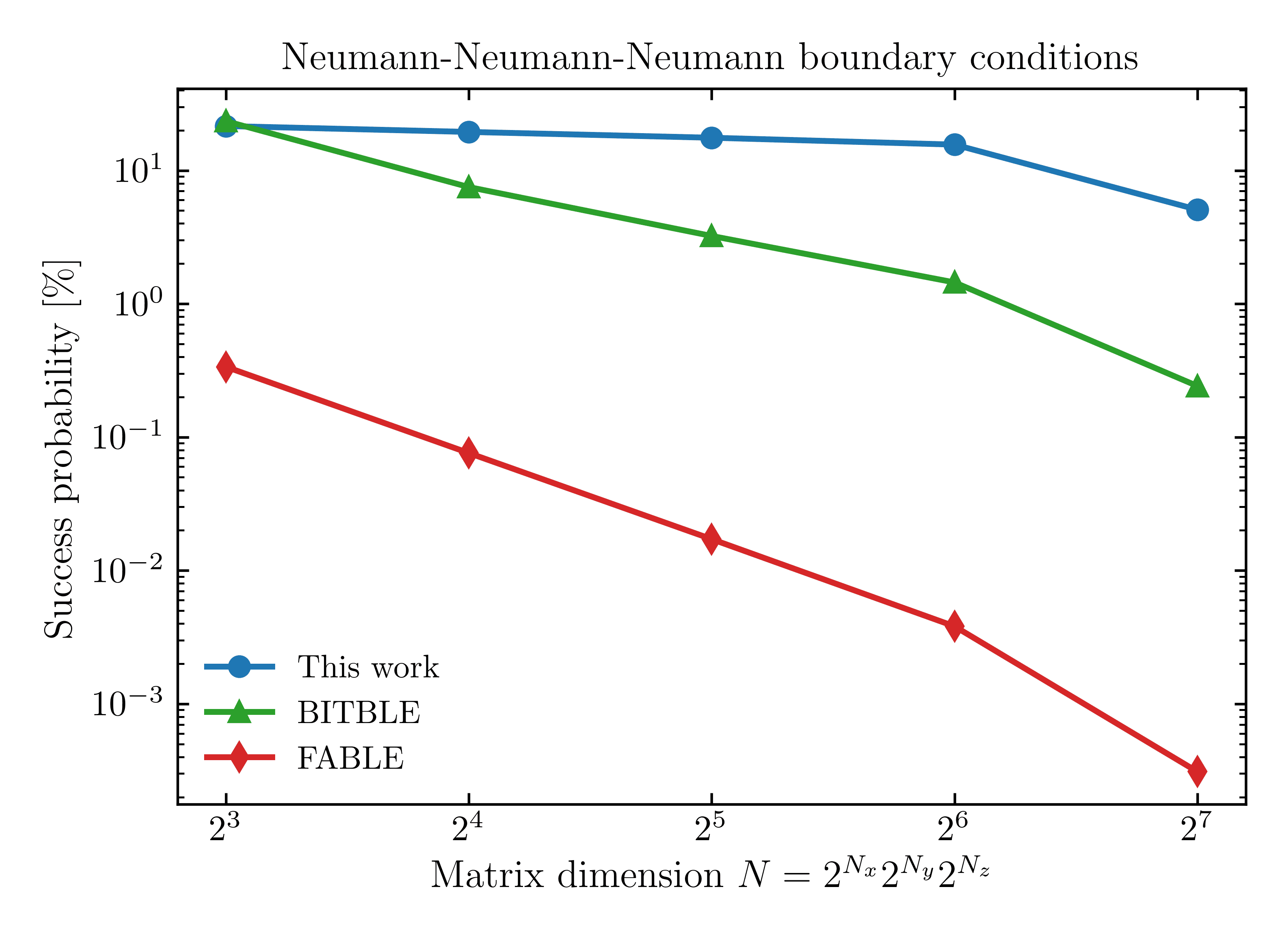}
\caption{Neumann-Neumann-Neumann}
\end{subfigure}
\caption{Success probability of the block encoding constructions for the $D$ dimensional Laplacian with various boundary conditions.}
\label{fig:SP_3D_DDD_NNN}
\end{figure}

\end{document}